\documentclass[11pt,twoside]{article}
\usepackage{sectsty}
\usepackage{tikz}
\usetikzlibrary{shapes}
\usepackage{natbib}
\usepackage{hyperref}
\usepackage{doi}
\usepackage{xurl}
\usetikzlibrary{intersections}   
\usetikzlibrary{arrows,decorations.markings}  
\usepackage[normalem]{ulem}
\usetikzlibrary{positioning}
\usetikzlibrary{arrows.meta}
\usepackage{pgfplots}
\pgfplotsset{compat=1.18} 
\usepackage{graphicx}
\usepackage{amsmath,amsthm}
\usepackage{geometry}
\usepackage{multirow}
\usepackage{float}
\usepackage{titlesec}
\usepackage{amssymb}

\usepackage[utf8]{inputenc}
\usepackage[T1]{fontenc}
\usepackage{lmodern}
\usepackage[english]{babel}
\usepackage{braket}
\usepackage{mathrsfs}
\usepackage{titleps}
\usepackage{tabularx}
\usepackage{enumitem}
\usepackage{wrapfig}
\usepackage{color}
\usepackage[perpage]{footmisc}
\usepackage[babel]{csquotes}
\usepackage{chngcntr}

\newcommand{\field}[1]{\mathbb{#1}}
\newcommand{\N}{\field{N}}
\newcommand{\R}{\field{R}}
\newcommand{\C}{\field{C}}
\newcommand{\Z}{\field{Z}}

\newcommand{\LL}{\mathcal L}

\newtheorem{theorem}{Theorem}[section]
\newtheorem{lemma}[theorem]{Lemma}
\newtheorem{proposition}[theorem]{Proposition}

\newtheorem{kor}[theorem]{Corollary}

\newcolumntype{L}[1]{>{\raggedright\arraybackslash}p{#1}}
\newcolumntype{C}[1]{>{\centering\arraybackslash}p{#1}}
\newcolumntype{R}[1]{>{\raggedleft\arraybackslash}p{#1}}

\titleformat{\paragraph}
{\normalfont\normalsize\bfseries}{\theparagraph}{1em}{}
\titlespacing*{\paragraph}
{0pt}{3.25ex plus 1ex minus .2ex}{1.5ex plus .2ex}



\addto{\captionsenglish}{
    
}

\sectionfont{\Large}
\subsectionfont{\Large}
\date{ }

\sectionfont{\Large}
\subsectionfont{\large}

\counterwithout{figure}{figurecount}

\begin{document}

\title{\Large{\textbf{RIGOROUS LOWER BOUNDS FOR ATOMIC AND IONIC GROUND-STATE ENERGIES}}}

\author{\normalsize{MICHAEL HOFACKER}\footnote{
Institute for Analysis, Karlsruhe Institute of Technology (KIT), D-76128 Karlsruhe, Germany}}

\maketitle

\setcounter{section}{0}

\begin{abstract}
We establish a new iterative method for computing rigorous lower bounds for the ground-state energies of all atoms and ions. The accuracy of the resulting bounds depends on two ingredients: the quality of the initial energy estimate and the quality of the used upper bound for the maximum number of bound electrons, $N_c(Z)$, as a function of the nuclear charge $Z>0$.
In particular, this new method allows us to establish the explicit lower bound $E(N,Z) \geq -0.8710\,N^{1/3}Z^2$ for the ground-state energy $E(N,Z)$ of the atom/ion with $N \in \N$ electrons and nuclear charge $Z \in (0,N]$, thereby considerably improving the best known constant $A=(3/2)^{1/3} \approx 1.1447$ in the estimate $E(N,Z) \geq -A\,N^{1/3}Z^2$ for $0 < Z \leq N$ and coming much closer to the best possible constant 
given by the Thomas-Fermi asymptotics $\lim_{N \to \infty} 
 |E(N,N)|N^{-7/3}= -E_{\textup{TF}}(1)\approx 0.7687$ \cite{LiebThomasFermi}. To demonstrate the potential and the limits of our method, we also compute explicit lower bounds for the ground-state energies $E(N,N)$ of all atoms with $2 \leq N \leq 86$ and we compare them with the bounds obtained in \cite{LiebHertel} and \cite{Seco}.
 Finally, we also provide a uniqueness result for the expectation values of the kinetic, potential and interaction energies on the ground-state space. 
\end{abstract}

\section{Introduction}
Since the early days of quantum mechanics, considerable effort has been devoted to determining the ground-state energies of atoms and molecules as accurately as possible, both experimentally and theoretically. High numerical accuracy is particularly important when differences of several energy levels are considered, as for example is the case of chemical reactions and radiative processes.
Furthermore, lower bounds for atomic and molecular ground-state energies play a fundamental role in the study of the stability of matter, since their asymptotic behavior as the number of electrons $N \in \N$ tends to infinity determines whether matter is stable or not \cite{StabilityofMatter}. Deriving an upper bound for a given atomic or molecular ground-state energy is relatively straightforward. Indeed, by the variational principle (cf. Eq.~\eqref{VarPrinciple} below), it suffices to compute the expectation value of the corresponding Hamiltonian with respect to an arbitrary trial state $\psi$ belonging to its domain. However, one can only expect that the resulting upper bound is close to the (typically unknown) exact value of the ground-state energy if $\psi$ provides a sufficiently good approximation of the true ground state. Over the years, a large toolbox of methods has been developed in quantum chemistry for constructing such approximate ground states. A prominent example is the self-consistent field (SCF) method based on the Hartree-Fock-Roothaan equations. Its accuracy can be further improved by employing Post-Hartree-Fock methods, which also account for electron correlation. For details concerning the SCF method and related approaches, we refer to standard textbooks on quantum chemistry (see, e.g., \cite{Lewars}).
The primary objective of this work is the derivation of lower bounds for atomic and ionic ground-state energies. From a mathematical perspective, this is a considerably more challenging problem, and correspondingly much less rigorous results are available. For small atoms and ions, highly accurate lower bounds have been obtained using, for example, the Bazley-Fox method \cite{Bazley}, the Temple inequality and related methods \cite{AlexanderColdwell, King, Stevenson1938}, as well as subsequent refinements of these approaches \cite{Ireland2022LowerBounds, Marmorino2012}. For larger systems, however, these methods become increasingly inefficient due to the growing dimensionality of the problem and the associated computational costs. In this regime, the approaches developed in \cite{LiebHertel} and \cite{Seco} yield stronger lower bounds, although these bounds still differ from the expected ground-state energies (e.g. upper bounds obtained from extensive numerical computations) by several percent. Moreover, both approaches rely on substantial numerical computations at intermediate stages, making it difficult to derive explicit bounds in terms of the electron number $N \in \N$. At the current state of research, the absence of sufficiently strong lower bounds for atomic and ionic ground-state energies seems to be one of the main obstacles to resolving several prominent open problems in mathematical physics, including the monotonicity of the ionization energies \cite{SimonProblems} and the derivation of the van der Waals law without additional assumptions \cite{Anapolitanos2016, AnapolitanosSigal2013}.

\section{Main results} \label{secmain}
We consider an atom/ion  with $N \in \N_0$ electrons and a nucleus of charge  $Z>0$ that is fixed at the origin in the Born-Oppenheimer approximation. In atomic units where $\hbar=1=4\pi \varepsilon_0=m=e$ the associated electronic Hamiltonian is given by
\begin{align}
H_N(Z):=\;\sum_{i=1}^{N} -\dfrac{1}{2}\Delta_{x_i} - \sum_{i=1}^N \frac{Z}{\left|x_i\right|} +  \sum_{1 \leq i< j \leq N} \frac{1}{\left|x_i-x_j\right|}\,. \label{DefHamiltonian}
\end{align}
Our main focus lies on the ground-state energy
\begin{align}
E(N,Z)= \inf \sigma(H_N(Z))\,. \label{DefE}
\end{align}
For $N=0$ we have that $E(0,Z)=0$ and for $N=1$ it is well-known that $E(1,Z)=-Z^2/2$, but for $N \geq 2$ the exact value of the negative number $E(N,Z)$ is not known in general and it depends essentially on the underlying Hilbert space $\mathcal{H}$ and the domain $D(H_N(Z))$ of the Hamiltonian.
As a consequence of the Pauli exclusion principle and the electron spin, the natural choice for the Hilbert space is the fermionic Hilbert  space
\begin{align*}
\mathcal{H}=\bigwedge_{i=1}^{N} (L^2(\R^3) \otimes \C^q)  \qquad \qquad \text{(Case FER)}
\end{align*}
with $q=2$ degrees of freedom in the spin.
Note that functions $\psi \in \mathcal{H}$ can be identified with 
complex-valued $L^2$-functions that are antisymmetric w.r.t. the $N$ position and spin coordinates $(x_i, s_i) \in \R^3 \times \left\{1,2\right\}$ of the electrons, i.e. for all $i\neq j$,
\begin{align*}
\psi(x_1,s_1,\ldots,x_i,s_i,\ldots,x_j,s_j,\ldots,x_N,s_N) = -\psi(x_1,s_1,\ldots,x_j,s_j,\ldots,x_i,s_i,\ldots,x_N,s_N)\,.
\end{align*}
But most of our results can be extended to a more general framework. For example, we could allow for 
arbitrary $q \in \N$ and also for different particle statistics like bosonic statistics or
even no symmetry at all. 
However, since it is well-known that the ground-state energy in the full-space case 
\begin{flalign*}
\mathcal{H}= L^2(\R^{3N}) \qquad \qquad \qquad  &\text{(Case FS)}
\end{flalign*}
coincides with the bosonic ground-state energy, we restrict ourselves to (Case FER) and (Case FS) in the following. In both of these cases it is well-known that $H_N(Z)$ is self-adjoint and bounded from below on the domain $D(H_N(Z))=\mathcal{H}^2$, where in (Case FS)
$\mathcal{H}^k = H^k(\R^{3N})$ just denotes the usual Sobolev space for $k \in \N$, and in (Case FER),
\begin{align*}
\mathcal{H}^k:= \left\{ \psi \in \mathcal{H} \,: \,  \forall s_1,\ldots,s_N \in \left\{1,\ldots,q \right\}: \psi(\, \cdot \,, s_1,\ldots, \, \cdot \,, s_N) \in H^k(\R^{3N}) \right\}\,.
\end{align*}
Here, $s_i \in \left\{1,\ldots,q\right\}$ denotes the spin coordinate of the $i$-th electron.  
Upper bounds for $E(N,Z)$ can be derived from the variational principle
\begin{align}
E(N,Z)= \inf_{\psi \in \mathcal{H}^2, \| \psi \|=1} \braket{\psi | H_N(Z) \psi}\,.
\label{VarPrinciple}
\end{align}
However, as already explained in the introduction, our focus lies on lower bounds for $E(N,Z)$, the main results being Theorems \ref{theo1} and \ref{theo2} below. But before we can state these results, we first have to introduce some notation and we also recall some related results that will be used in the following.

We start by recalling that the essential spectrum of $H_N(Z)$ can be characterized by the HVZ theorem (see e.g. \cite[Theorem 11.2]{Teschl})
\begin{align}
\sigma_{\textup{ess}}(H_N(Z))= [E(N-1,Z), \infty)\,. \label{HVZ}
\end{align}
An important consequence of Eq. \eqref{HVZ} is that the strict inequality $E(N,Z)<E(N-1,Z)$ implies that $E(N,Z)$ is a discrete eigenvalue  of $H_N(Z)$, so the set of all normalized ground states
\begin{align}
\text{GS}_N(Z):=\left\{ 
\psi \in \mathcal{H}^2 \;:\; H_N(Z)\psi = E(N,Z)\psi \text{ and } \|\psi\|=1 \right\}  \label{DefGS}
\end{align}
is non-empty whenever $E(N,Z)<E(N-1,Z)$. However, a nucleus with a fixed charge number $Z>0$ can only bind finitely many electrons (the first rigorous proofs for this were given independently by Ruskai \cite{Ruskai1, Ruskai2} and Sigal \cite{Sigal1, Sigal2}). Explicitly, we define the maximum number of bound electrons as
\begin{align}
N_c(Z):= \max\left\{N \in \N \,: \, E(N,Z) < E(N-1,Z) \right\}\,.
\label{DefNc}
\end{align}
To explain the motivation behind this definition, we first note that adding more electrons can never increase
the ground-state energy, i.e.
\begin{align}
E(N,Z)\leq E(N-1,Z) \qquad \qquad (N \in \N, Z>0)\,. \label{Edecrease}
\end{align}
This is due to the fact that one can always send one electron to infinity and simultaneously make its kinetic energy arbitrarily small. From Eqs. \eqref{DefNc} and \eqref{Edecrease} we conclude that $N_c(Z)$ is the maximum number of electrons that can be attached to the nucleus with a decrease in the ground-state energy. In particular, one has that $E(N,Z)=E(N_c(Z),Z)$ for all $N \geq N_c(Z)$.

We now recall some explicit bounds for $N_c(Z)$. One the one hand, Zhislin \cite{Zhislin1, Zhislin2} has shown that for $Z>N-1$ the Hamiltonian $H_N(Z)$ has infinitely many discrete eigenvalues below the bottom of the essential spectrum, which implies that $N_c(Z) \geq Z$. 
One the other hand, in (Case FER) with $q=2$ one expects that 
\begin{align}
N_c(Z) \leq Z + C  \label{IonizationConjecture}
\end{align} 
for some constant $C \geq 1$ that is independent of $Z\geq 1$. This is known as the ionization conjecture. 
It has been verified in some approximate models like the Thomas-Fermi model \cite{LiebThomasFermi} or the Hartree-Fock model \cite{SolovejHartreeFock}. Very recently, an AI-based proof has been published \cite{OpenAI2026UniformExcessCharge} - however, without giving an explicit upper bound for the constant\footnote{Tracking the contributions to the constant $C$ with the help of AI leads to $C \approx 2^{10^{18}}$, which is considerably larger than the number of electrons in the observable universe ($\approx 10^{80}$).}, which is expected to be $C=2$ or even $C=1$.
Before only the weaker asymptotic neutrality
\begin{align}
\lim_{Z \to \infty} \dfrac{N_c(Z)}{Z}=1 \label{AsmptoticNeutrality}
\end{align}
has been rigorously established \cite{AsymptoticNeutrality1, AsymptoticNeutrality2}.
We emphasize that Eq. \eqref{AsmptoticNeutrality} is \textit{not} valid for bosons because for bosons this limit is known to be strictly larger than one \cite{LiebBosons, SolovejBosons}. 
Soon after Sigal's and Ruskai's pioneering works \cite{Ruskai1, Ruskai2, Sigal1, Sigal2} Lieb \cite{LiebIonization} gave an elegant proof of the bound
\begin{align}
N_c(Z) < 2Z +1 \qquad \qquad (Z>0)\,, \label{LiebBound}
\end{align}
which holds irrespectively of the particle statistics and the number $q \in \N$ of degrees of freedom in the spin. In (Case FER) with $q=2$, Nam \cite{NamIonization} has refined Lieb's method to establish 
\begin{align}
N_c(Z) < 1.22Z +3 Z^{1/3} \,, \label{NamBound}
\end{align}
which improves Estimate \eqref{LiebBound} for $Z \geq 5.52$.
Nam's proof has been further improved by Hundertmark, Pattakos and Schulz \cite{HundertmarkExcesscharge}, and in particular it has been shown that
\begin{align}
N_c(Z) < 1.1185Z +4 Z^{1/3} \,, \label{SchulzBound}
\end{align}
which improves Estimate \eqref{NamBound} for $Z \geq 30.93$.
In any case, we may consider a function $f: (0,\infty) \rightarrow [1,\infty)$ satisfying
\begin{align}
N_c(Z) < f(Z)\qquad \qquad (Z>0)\,. \label{Ncbound}
\end{align}
For  simplicity, we further assume that $f$ is strictly monotonically increasing and that $f^{-1}(n)$ exists for all sufficiently large $n \in \N$.   
In (Case FER) with $q=2$, it follows from Estimates \eqref{LiebBound}-\eqref{SchulzBound} that Estimate \eqref{Ncbound} is satisfied for
\begin{align}
f(Z):=  \begin{cases}
2Z+1 & \qquad \textup{if} \; Z \leq 5.5\,, \\
1.22Z +3 Z^{1/3} & \qquad \textup{if} \; 5.5 < Z < 30.9\,, \\
1.1185Z +4 Z^{1/3} & \qquad \textup{if} \; Z \geq 30.9\,.
\end{cases} \label{fchoice}
\end{align}
Besides the ionization bound $f$, another key element of Theorem \ref{theo1} below is the absolute value of the ground-state energy when $N$ electrons are attached to a nucleus of charge $Z=1$ and the electron-electron repulsion is completely neglected, which we denote by $\sigma_N$. Explicitly,
\begin{align}
\sigma_N = |\inf \sigma(h_N)|\,,  \label{DefSigmaN}
\end{align}
where 
\begin{align}
h_N:=\;\sum_{i=1}^{N} \left( -\dfrac{1}{2}\Delta_{x_i} - \frac{1}{\left|x_i\right|} \right) \label{Defh_N}
\end{align}
denotes the Hamiltonian without the electron-electron repulsion. 
In (Case FS) it is easy to see that $\sigma_N=N/2$, while in (Case FER) $\sigma_N$ is the sum of the absolute values of the first $N$ eigenvalues of the hydrogen atom counting multiplicity (if each electron has $q$ degrees of freedom in the spin). For details concerning $\sigma_N$, we refer to Appendix \ref{secsigma_N}. With these preparations at hand, we  now state our first result:

\begin{theorem} \label{theo1}
Let $N_0 \in \N$, $Z_0 \in (0,f^{-1}(N_0+1)]$ and $E_0 \in (-\infty, E(N_0,Z_0)]$, where the  strictly monotonically increasing function $f$ satisfies Estimate \eqref{Ncbound} and $f^{-1}(n)$ exists for all $n \in \N$ with $n \geq N_0 +1$. For $N \in \N$ with $N \geq N_0$ and $Z \in [Z_0, \infty)$, we then define $L(N,Z)$ recursively by 
\begin{align}
L(N_0,Z):= \sqrt{\left| E_0 \right|} 
+ \sqrt{\sigma_{N_0}} (Z-Z_0) \label{DefL0}
\end{align}
and, for $N \geq N_0+1$,
\begin{align}
L(N,Z):= \begin{cases}
L(N-1,Z) & \qquad \textup{if} \; Z_0 \leq Z \leq f^{-1}(N)\,, \\
L(N-1,f^{-1}(N)) + \sqrt{\sigma_{N}} (Z-f^{-1}(N))  & \qquad \textup{if}\; Z>f^{-1}(N)\,.
\end{cases} \label{DefL}
\end{align}
Then $E_{L}(N,Z):= -L(N,Z)^2$ is a lower bound for $E(N,Z)$, i.e.
\begin{align*}
E(N,Z) \geq  E_{L}(N,Z) \qquad \qquad (N \geq N_0, Z \geq Z_0)\,.
\end{align*}
\end{theorem}

Once one has an a priori lower bound for $E(N_0,Z_0)$ for some $N_0 \in \N$, $Z_0 \in (0,f^{-1}(N_0+1)]$, Theorem \ref{theo1} allows one to compute recursively the lower bound $E_L(N,Z)$ for all $N \geq N_0$ and $Z \geq Z_0$. 
In Section \ref{secdiscussion} we compute $E_L(N,N)$  for all atoms with atomic number $2 \leq N \leq 86$ for two choices of the initial values $N_0, Z_0, E_0$ and the ionization bound $f$, and we compare our results with the lower and upper bounds from references \cite{LiebHertel} and \cite{Seco}. 
The lower bounds in \cite{LiebHertel} and \cite{Seco} are based on a lower bound for the electron-electron repulsion in terms of an effective one-particle potential $V_{\textup{eff}}$. More precisely, for some constant $\Lambda>0$,
\begin{align*}
 \sum_{1 \leq i< j \leq N} \frac{1}{\left|x_i-x_j\right|} \geq 
 \sum_{i=1}^{N} V_{\textup{eff}}(x_i) - \Lambda \qquad \qquad (x_i \neq x_j\; \text{for}\; 1 \leq i < j \leq N)\,.
\end{align*}
This reduces the problem of finding   a lower bound for $E(N,N)$ to computing the eigenvalues of the one-particle operator $
h_{\textup{eff}}= -\frac{1}{2}\Delta - N/|x|  + V_{\textup{eff}}$
and summing them up appropriately afterwards. In spite of the fact that the eigenvalues of $h_{\textup{eff}}$ can be computed with numerical methods, one of the main advantages of our method is that it is \textit{completely analytical} since we do not rely on \textit{any}  numerical computations. 
Even more remarkably, we can derive from Theorem \ref{theo1} an explicit $N$- and $Z$-dependent lower bound for the ground-state energy of all atoms and negatively charged ions:

\begin{theorem}\label{theo2}
Consider (Case FER) with $q=2$. Then, for $\beta=1.1185$ and any $\eta \in [1/\beta, \infty)$, the ground-state energy $E(N,Z)$ with $N \in \N$ and $Z \in (0,\eta N]$ has the lower bound
\begin{align}
E(N,Z) \geq - \alpha(\eta, \beta)^2 \, N^{1/3} Z^2\,, \label{Elower}
\end{align}
where
\begin{align}
\alpha(\eta, \beta):= \dfrac{2^{5/6}\cdot 3^{7/6}}{7\eta \beta}    + \left(\dfrac{3}{2}\right)^{1/6}\left( 1 - \frac{1}{\eta \beta} \right)\,.  \label{Defalpha}
\end{align}
In particular, 
\begin{align}
E(N,Z) \geq -0.8710\,N^{1/3} Z^2 \qquad \qquad (Z \leq N \in \N) \label{IonEstimate}
\end{align}
and
\begin{align}
E(N,N) \geq -0.8710\,N^{7/3} \qquad \qquad (N \in \N)\,. \label{AtomEstimate}
\end{align}
\end{theorem}
\noindent
\textit{Remarks.} 
\begin{enumerate}[label=\arabic*.)]
\item 
To prove Estimate \eqref{Elower}, we estimate the lower bound from Theorem \ref{theo1} (with suitably chosen $N_0, Z_0, E_0, f$) in terms of $N \in \N$ and $Z \in (0,\eta N]$. But without asking for a closed formula, Theorem \ref{theo1} gives a better lower bound for $E(N,Z)$  than Estimate \eqref{Elower}.

\item Note that  $\alpha(\,\cdot\,, \beta)$ is strictly monotonically increasing with $\lim_{\eta \to \infty} \alpha(\eta, \beta)=(3/2)^{1/6}$. Therefore, Estimate \eqref{Elower} implies the well-known estimate
\begin{align}
E(N,Z) \geq - \left( 3/2 \right)^{1/3} \, N^{1/3} Z^2  \qquad \qquad (N \in \N, Z>0)\,.\label{Bohrbound}
\end{align} 
This has been used in the proofs of several other results, in particular \cite[Theorem 1]{NamIonization} (see \cite[Eq. (10)]{NamIonization}) and \cite[Proposition 2.5]{HundertmarkExcesscharge} (see \cite[Lemma A.7]{HundertmarkExcesscharge}).
However, for given $N \in \N$ and $Z>0$ our Estimate \eqref{Elower} gives a better lower bound for $E(N,Z)$. 
Therefore, \cite[Theorem 1]{NamIonization} and \cite[Proposition 2.5]{HundertmarkExcesscharge} can be slightly improved by using Estimate \eqref{Elower} instead of Estimate \eqref{Bohrbound}, but this only affects the correction terms of order $Z^{1/3}$ and not the leading terms of order $Z$. Since \cite[Theorem 1]{NamIonization} and \cite[Proposition 2.5]{HundertmarkExcesscharge} are in turn used in the proof of Estimate \eqref{Elower}, these results allow to improve each other iteratively\footnote{The relative improvement becomes vanishingly small after very few iterations.}.

\item By \cite[Theorem III.1]{LiebThomasFermi}, the leading term of $E(N,N)$ in the limit $N \to \infty$ coincides with the Thomas-Fermi energy $E_{\textup{TF}}(1)$ (numerically one finds that $E_{\textup{TF}}(1) \approx -0.7687$), namely
\begin{align}
\lim_{N \to \infty} E(N,N)N^{-7/3}= E_{\textup{TF}}(1)\,. \label{GroundstateLimit}
\end{align}
For a more refined asymptotic expansion of $E(N,N)$ as $N \to \infty$, we refer to \cite{FeffermanDirac} and the references therein. In particular, Eq. \eqref{GroundstateLimit} shows that the best possible constant in the estimate 
$E(N,N) \geq -AN^{7/3}$ would be $A=|E_{\textup{TF}}(1)|$. 
Although our proof does not yield this result, the constant in Theorem \ref{theo2}  could be improved to $A=0.8411$ if the ionization conjecture \eqref{IonizationConjecture} were proven with a sufficiently good upper bound for the constant $C$.
\item The monotonicity of the ionization energies and electron affinities of an atom is an outstanding problem in mathematical physics. For details, we refer to Appendix \ref{secionization}, where we show that the analog property holds for the lower bounds $E_L$ from Theorem \ref{theo2}. 
\end{enumerate}

An explicit lower bound for $E(N,Z)$ that only depends on $N$  and $Z$ can be derived from Theorem \ref{theo2} 
by choosing  $\eta:= \max(Z/N,1/\beta)$ and using that $\alpha(\,\cdot\,, \beta)$ is strictly monotonically increasing: 

\begin{kor}\label{korbound}
Consider (Case FER) with $q=2$. Then, for all $N \geq 2$ and $Z>0$ the ground-state energy $E(N,Z)$ has the lower bound
\begin{align*}
E(N,Z) \geq -\max((C_1 Z -C_2N)^2, C_3 Z^2) \,N^{1/3},  
\end{align*}
where 
\begin{align*}
C_1&=(3/2)^{1/6} \approx 1.07\,, \\
C_2&= ((3/2)^{1/6} - 2^{5/6} 3^{7/6}/7)/1.1185 \approx 0.1367\,,\\
C_3&= (2^{5/6} 3^{7/6}/7)^2\approx 0.8410\,.
\end{align*}
\end{kor}

\noindent\textit{Remark.} A similar lower bound for the ground-state energy $E_{\textup{mol}}(N,\underline{Z}, \underline{R})$ of a molecule with $N \in \N$ electrons and $M \in \N$ nuclei with charges $\underline{Z}=(Z_1,\ldots,Z_M) \in (0,\infty)^M$ and nuclear positions $\underline{R}=(R_1,\ldots,R_M) \in \R^{3M}$ can be derived 
with the help of the kinetic Lieb-Thirring inequality in combination with the exchange estimate from  \cite[Theorem 6.1]{StabilityofMatter}. 
Explicitly, it follows from \cite[Estimate (7.1.14)]{StabilityofMatter} that
\begin{align}
E_{\textup{mol}}(N,\underline{Z}, \underline{R}) \geq  
- \dfrac{2^{5/3}}{K} N\left(0.84 + \max_{k=1,\ldots,K}(Z_k) \sqrt{D} (M/N)^{1/3} \right)^2 \label{Emollower}
\end{align}
uniformly in $\underline{R} \in \R^{3M}$, where $D=(9/10)2^{-1/3}\pi^{4/3} \approx 3.287$ and $K$ denotes the kinetic Lieb-Thirring constant. We note that in the atomic/ionic case $M=1$ Corollary \ref{korbound} always yields a better lower bound than Estimate \eqref{Emollower}, even if the kinetic Lieb-Thirring inequality were proven with the classical constant $K=K_{\textup{cl}}=(3/5)(6 \pi^2)^{2/3}$.\medskip \newline

In (Case FS) it is well-known that if $E(N,Z)$ is an eigenvalue of $H_N(Z)$, then $E(N,Z)$ has multiplicity one, which is known as non-degeneracy of the ground state (see e.g. \cite[Theorem 10.12]{Teschl}). Moreover, the corresponding eigenfunction can be chosen to be strictly positive and symmetric w.r.t. permutations of the particle positions. Combining this with Zhislin's result \cite{Zhislin1, Zhislin2}, it follows that $H_N(Z)$ has a non-degenerate ground state if $Z>N-1$. But in the presence of non-trivial permutation symmetry and spin, non-degeneracy of the ground state is not expected in general and only very few rigorous results are known. For example, it is known that in (Case FER) with $q=2$, the helium atom ($N=Z=2$) has a non-degenerate ground state. This also follows from \cite[Theorem 10.12]{Teschl} because in this case it is well-known that the ground state  is a product of the (symmetric) ground state from (Case FS) and an antisymmetric spin function. However, independent of the non-degeneracy of the ground state, we have the following general result for the expectation values of the kinetic energy, the potential energy and the interaction energy on the ground-state space:

\begin{theorem} \label{theo3}
For fixed $N \in \N$, the function $E(N, \,\cdot\,): (0,\infty) \rightarrow \R, Z \mapsto E(N,Z)$ is strictly monotonically decreasing, concave and differentiable in all but perhaps countably many points. If $Z>0$ is chosen so that $E(N,Z)$ is an eigenvalue of $H_N(Z)$ and 
$E(N, \,\cdot\,)$ is differentiable at $Z$, then we have the following virial type identities for all $\psi \in \text{GS}_N(Z)$:
\begin{align*}
\Braket{\psi|T(N)\psi} &= \left|E(N,Z)\right|\, , \\
\Braket{\psi|V(N)\psi} &= -\dfrac{\partial E}{\partial Z}(N,Z)\,, \\
\Braket{\psi|I(N)\psi} &= 2 E(N,Z) - Z\dfrac{\partial E}{\partial Z}(N,Z) \,,
\end{align*}
where $T(N)$, $V(N)$ and $I(N)$ are defined in Eq. \eqref{DefTVW} below. In particular, this implies that all $\psi \in \text{GS}_N(Z)$ have the same expectation values for the kinetic energy, the potential energy and the interaction energy, respectively.
\end{theorem}

Note that the one-sided derivatives $\frac{\partial E}{\partial Z+}(N,Z)$ and $\frac{\partial E}{\partial Z-}(N,Z)$ exist for all $Z>0$ because $E(N,Z)$ is concave in $Z$. The following characterization 
of these one-sided derivatives  shows that the assumption in Theorem \ref{theo3} that $E(N, \,\cdot\,)$ is differentiable at $Z$ is necessary:
 
\begin{theorem} \label{thmEderivative}
Let $N \in \N$, $Z>0$ and suppose that  
 $E(N,Z)<E(N-1,Z)$. Then, 
\begin{align*}
\frac{\partial E}{\partial Z_{+}}(N,Z)=\min_{\psi \in \text{GS}_N(Z)} - \Braket{\psi|V(N)\psi}, \qquad
\frac{\partial E}{\partial Z_{-}}(N,Z)=\max_{\psi \in \text{GS}_N(Z)} - \Braket{\psi|V(N)\psi}\,.
\end{align*}
\end{theorem}

The outline of this paper is as follows: In Section \ref{secdiscussion} we compute the lower bounds $E_L(N,N)$ from Theorem \ref{theo1} for all atoms with atomic number $2 \leq N \leq 86$ for two  choices of the initial values $N_0, Z_0, E_0$ and the ionization bound $f$, and we compare them with the lower and upper bounds from references \cite{LiebHertel} and \cite{Seco}. In Section \ref{sectheo13} we give the proofs of Theorems \ref{theo1}, \ref{theo3} and \ref{thmEderivative}. Finally, Section \ref{sectheo2} is dedicated to the proof of Theorem \ref{theo2}. 

\section{Discussion of Theorem \ref{theo1} and comparison with other results} \label{secdiscussion}
Throughout this section we consider (Case FER) with $q=2$. We  are going to compare the lower bounds resulting from Theorem \ref{theo1} with the lower and upper bounds from references \cite{LiebHertel} and \cite{Seco} for all atoms with atomic number $2 \leq N \leq 86$. For the initial values $N_0$, $Z_0, E_0$ and the ionization bound $f$ in Theorem \ref{theo1} that satisfies Estimate \eqref{Ncbound}, we consider two choices. 
Our first choice is the ionization bound $f$ defined by Eq. \eqref{fchoice}. For the initial energy estimate we make use of the lower bound 
$E(2,Z) \geq  \min(-Z^2+0.4226\,Z, -5Z^2/9)$ from \cite[Equation (4)]{Briet}. This leads to the choice
\begin{flalign}
1.)& \quad N_0=2, \quad Z_0=f^{-1}(3)=1, \quad E_0=-0.5774 = \min(-1+0.4226, -5/9)\leq E(2,1) \,.& \label{firstchoice}
 \end{flalign}
The resulting lower bounds $E_L(N,Z)$ for $2 \leq N \leq 7$ are illustrated in Figure \ref{E_Lgraphs}, below.

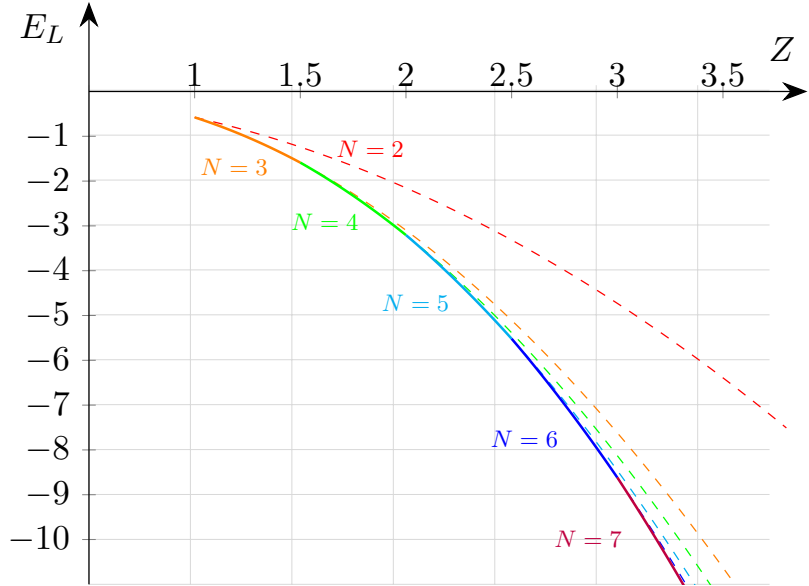
\begin{figure}[H] 
\begin{center}
\begin{tikzpicture}[scale=1.2]

\newcommand{\Zstart}{1}
\newcommand{\Nstart}{2}
\newcommand{\E}{0.5774^(1/2)}

\draw[xstep=1.118cm, ystep=0.495cm,gray!30,very thin] (0,0) grid (7.5,5.5);

\begin{axis}[width=9.5cm, height=8.0cm,
    trig format plots=rad,  axis lines=middle,  
    samples=300,
        declare function={
      func(\x)={-(\E+(0.5)^0.5*(\x-\Zstart))^2};
      gunc(\x)={-((-func(1.0))^0.5+(1)^0.5*(\x-1))^2};
      hunc(\x)={-((-gunc(1.5))^0.5+(1.125)^0.5*(\x-1.5))^2};
      iunc(\x)={-((-hunc(2.0))^0.5+(1.25)^0.5*(\x-2))^2};
      junc(\x)={-((-iunc(2.5))^0.5+(1.375)^0.5*(\x-2.5))^2};
      kunc(\x)={-((-junc(3.0))^0.5+(1.5)^0.5*(\x-3.0))^2};
    },
   xlabel={$Z$},
   xlabel style={xshift=0.0cm, yshift=0.2cm},
   ylabel={$E_L$},
   ylabel style={xshift=-0.9cm, yshift=0.0cm},
   color=black,
    xmin=0.5, xmax=3.9, ymin=-11.0, ymax=2.00,
    xtick={-1,-2},
xtick={1,1.5,2,2.5,3,3.5},
 x tick label style={above},
 ytick={-10,-9,-8,-7,-6,-5,-4,-3,-2,-1},
  axis line style = {-Stealth[{length=8pt}]} 
   ]

\addplot [
    no markers,
    red,
    dashed,
    domain=1:3.8
] {func(x)};

\addplot [
    no markers,
    orange,
    thick,
    domain=1:1.5
] {gunc(x)};

\addplot [
    no markers,
    orange,
    dashed,
    domain=1.5:3.8
] {gunc(x)};

\addplot [
    no markers,
    green,
    thick,
    domain=1.5:2
] {hunc(x)};

\addplot [
    no markers,
    green,
    dashed,
    domain=2:3.6
] {hunc(x)};

\addplot [
    no markers,
    cyan,
    thick,
    domain=2.0:2.5
] {iunc(x)};

\addplot [
    no markers,
    cyan,
    dashed,
    domain=2.5:3.6
] {iunc(x)};

\addplot [
    no markers,
    blue,
    thick,
    domain=2.5:3.0
] {junc(x)};

\addplot [
    no markers,
    blue,
    dashed,
    domain=3.0:3.6
] {junc(x)};

\addplot [
    no markers,
    purple,
    thick,
    domain=3.0:3.5
] {kunc(x)};

\end{axis}

\node[color=red] at (3.1,4.8) {\footnotesize $N=2$};
\node[color=orange] at (1.6,4.6) {\footnotesize $N=3$};
\node[color=green] at (2.6,4.0) {\footnotesize $N=4$};
\node[color=cyan] at (3.6,3.1) {\footnotesize $N=5$};
\node[color=blue] at (4.8,1.6) {\footnotesize $N=6$};
\node[color=purple] at (5.5,0.5) {\footnotesize $N=7$};
\end{tikzpicture}
\caption{The lower bound $E_L(N,Z)$ for the initial values from Eq. \eqref{firstchoice} and the function $f$ from Eq. \eqref{fchoice} (i.e. $f(Z)=2Z+1$ for $Z \leq 5.5$).
The dashed curves of different colors correspond to different values of $N$,
but they all have the same starting point $E_L(N,Z=1)=E_0=-0.5774$. The union of the solid lines corresponds to the function $Z \mapsto E_L(\lfloor f(Z) \rfloor ,Z)$, where the floor function $\lfloor \cdot \rfloor$  is defined by $\lfloor x \rfloor := \max\left\{ n \in \Z \, : \, n \leq x  \right\}$ for $x \in \R$. From the definition of $E_L$ in Theorem \ref{theo1}, it follows that $E_L(\lfloor f(Z) \rfloor ,Z)=E_L(7,Z)$ for $Z \leq 3.5$.}  \label{E_Lgraphs}
\end{center}
\end{figure}

In Table 1 below, we compare the lower bounds $E_L(N,N)$ resulting from the initial values from Eq. \eqref{firstchoice} and the ionization bound $f$ from Eq.  \eqref{fchoice} with the lower and upper bounds from references \cite{LiebHertel} and \cite{Seco} for $2 \leq N \leq 86$. It turns out that for $N=6$ our lower bound $E_L(6,6)$ is already more than $10$ percent worse than the lower bound $E_{lb}(N)$ from \cite{Seco}. This deviation increases with $N$ until for $N=86$ our lower bound $E_L(N,N)$ deviates more than $20$ percent from the lower bound from \cite{Seco}. However, it is important to note that 
our method allows for substantial improvement since experimental observations and numerical computations suggest that both the initial values $N_0, Z_0, E_0$ from Eq. \eqref{firstchoice} and the ionization bound $f$ from Eq. \eqref{fchoice} are far from being optimal.  Our desired choice for $N_0$ and $Z_0$ would be $N_0=2$ and $Z_0=Z_c$, where $Z_c$ is the critical charge number for which only one electron is bound to the nucleus, i.e.
\begin{align*}
Z_c:=\sup\left\{Z>0 \, : \, N_c(Z)=1 \right\} \,.
\end{align*}
For details concerning $Z_c$, we refer to Lemma \ref{lmZc} in the appendix.
In particular, we then have that 
$E(2,Z)=E(1,Z)=-Z^2/2$ for all $Z \in (0,Z_c]$ and 
$E(2,Z)<E(1,Z)$ for all $Z>Z_c$. The extensive numerical computations in \cite{Estienne2014} suggest that $Z_c = 0.91102822407725573(4)$, so we expect that at least $Z_c > Z_0:=0.911$ and thus $E(N_0,Z_0)=E(2,0.911)=E(1,0.911)=-(0.911)^2/2$.  
It remains to find a strictly monotonically increasing function $f$ that satisfies $N_c(Z)<f(Z)$, and which preferably does not deviate much from  $N_c(Z)$.
As already explained in Section \ref{secmain}, it is conjectured that one can choose $f(Z)=Z+C$ for $Z \geq 1$, where $C=2$ or even $C=1$ (ionization conjecture). 
Note that the constant $C$ can not be less than one since Bethe has shown that the negative hydrogen ion $\textup{H}^{-}$ 
exists \cite{Bethe} and hence $N_c(1) \geq 2$ \footnote{By Lieb's bound \eqref{LiebBound} one has that $N_c(1) = 2$.}. In summary, our second choice is therefore $f(Z)=Z+1$ and
\begin{flalign*}
2.)& \quad N_0=2, \qquad Z_0=0.911, \qquad E_0=-(0.911)^2/2 =E(N_0,Z_0)\,.&
\end{flalign*}
Under the above hypotheses, it is straightforward to check that all assumptions of Theorem \ref{theo1} are satisfied, so Theorem \ref{theo1} yields a second series of lower bounds $E_L^{*}(N,N)$ for $2\leq N \leq 86$ (cf. the second column in Table 1, below). These lower bounds improve the lower bounds from \cite{Seco} for $N \leq 12$ and the lower bounds from \cite{LiebHertel} for $N \in \left\{10,20,40\right\}$. However, for larger values of $N$ the lower bounds from \cite{LiebHertel} and \cite{Seco} are still better (but the relative deviation stays below $7$ percent for $N \leq 86$). Of course, one should also not forget that $E_L^{*}(N,N)$ would only become  a rigorous lower bound for $E(N,N)$ if the ionization conjecture with $C=1$  and the bound $Z_c > Z_0=0.911$ were rigorously proven.

\begin{figure}[H]
\begin{center}
\begin{tabular}{ |C{0.5cm}| C{2.0cm}| C{2.0cm}| C{2.2cm} | C{2.0cm}| C{2.2cm}| C{2.0cm}|}
\hline
$N$ & $E_L(N,N)$ & $E_L^{*}(N,N)$ & $E_{LB}(N)$ \cite{LiebHertel} & $E_{lb}(N)$ \cite{Seco} & $E_{UB}(N)$ \cite{LiebHertel} & $E_{ub}(N)$ \cite{Seco} \\
\hline
$2$ & $-3.10$ & $-3.01$ & - & - & - & - \\
$3$ & $-8.31$ & $-7.81$ & - & $-8.33$ & - & $-7.43$ \\
$4$ & $-16.69$ & $-15.31$& - & $-16.06$ & - & $-14.57$ \\
$5$ & $-28.81$ & $-25.86$ & - & $-26.76$ & - & $-24.53$ \\
$6$ & $-45.19$ & $-39.81$ & - & $-40.86$ & - & $-37.68$ \\
$7$ & $-66.38$ & $-57.52$ & - & $-58.73$ & - & $-54.40$ \\
$8$ & $-92.90$& $-79.34$ & - & $-80.75$ & - & $-74.81$ \\
$9$ & $-125.29$ & $-105.60$ & - & $-107.33$ & - & $-99.40$ \\
$10$ & $-164.09$ & $-136.67$ & $-162.66$ & $-138.90$ & $-127.97$ & $-128.54$ \\
$11$ & $-205.66$ & $-172.25$ & - & $-173.58$ & - & $-161.85$ \\
$12$ & $-252.80$ & $-212.50$ & - & $-212.92$ & - & $-199.61$ \\
$13$ & $-305.65$ & $-257.58$ & - & $-256.90$ & - & $-241.87$ \\
$14$ & $-364.41$ & $-307.65$ & - & $-305.86$ & - & $-288.85$ \\
$15$ & $-429.29$ & $-362.87$ & - & $-359.78$ & - & $-340.70$ \\
$16$ & $-500.51$ & $-423.40$ & - & $-418.90$ & - & $-397.50$ \\
$17$ & $-578.25$ & $-489.39$ & - & $-483.60$ & - & $-459.50$ \\
$18$ & $-662.74$ & $-561.01$ & - & $-553.62$ & - & $-526.81$ \\
$19$ & $-754.16$ & $-638.41$ & - & $-628.22$ & - & $-599.16$ \\
$20$ & $-852.73$ & $-721.75$ & $-788.27$ & $-707.65$ & $-673.17$ & $-676.75$ \\
$21$ & $-958.64$ & $-811.18$ & - & $-793.57$ & - & $-759.73$ \\
$22$ & $-1072.09$ & $-906.87$ & - & $-885.37$  & - & $-848.40$ \\
$23$ & $-1193.28$ & $-1008.96$ & - & $-982.98$ & - & $-942.88$ \\
$24$ & $-1322.42$ & $-1117.62$ & - & $-1086.21$ & - & $-1043.30$ \\
$25$ & $-1459.70$ & $-1233.00$ & - & $-1196.55$ & - & $-1149.86$ \\
$26$ & $-1605.31$ & $-1355.25$ & - & $-1313.57$ & - & $-1262.44$ \\
$27$ & $-1759.45$ & $-1484.54$ & - & $-1436.21$ & - & $-1381.41$ \\
$28$ & $-1922.32$ & $-1621.01$ & - & $-1565.08$ & - & $-1506.87$ \\
$29$ & $-2086.73$ & $-1764.24$ & - & $-1702.10$ & - & $-1638.95$ \\
$30$ & $-2258.91$ & $-1914.31$ & - & $-1854.83$ & - & $-1777.84$ \\
$31$ & $-2438.99$ & $-2071.31$ & - & $-2002.95$ & - & $-1923.26$ \\
$32$ & $-2627.06$ & $-2235.35$ & - & $-2156.22$ & - & $-2075.36$ \\
$33$ & $-2823.26$ & $-2406.50$ & - & $-2318.33$ & - & $-2234.24$ \\
$34$ & $-3027.67$ & $-2584.85$ & - & $-2487.53$ & - & $-2399.86$ \\
$35$ & $-3240.42$ & $-2770.49$ & - & $-2665.08$ & - & $-2572.44$ \\
$36$ & $-3461.63$ & $-2963.52$ & - & $-2848.65$ & - & $-2752.05$ \\
$37$ & $-3691.39$ & $-3164.03$ & - & $-3037.25$ & - & $-2938.35$ \\
$38$ & $-3929.82$ & $-3372.09$ & - & $-3237.45$ & - & $-3131.54$ \\
$39$ & $-4177.03$ & $-3587.81$ & - & $-3443.75$ & - & $-3331.68$ \\
\hline
\end{tabular}
\end{center}
\end{figure}

\newpage
\setcounter{figure}{0}
\renewcommand*{\figurename}{Table}
\begin{figure}[H]
\begin{center}
\begin{tabular}{ |C{0.8cm}| C{2.0cm}| C{2.0cm}| C{2.0cm}| C{2.0cm}| C{2.0cm}| C{2.0cm}|}
\hline
$N$ & $E_L(N,N)$ & $E_L^{*}(N,N)$ & $E_{LB}$ \cite{LiebHertel} & $E_{lb}$ \cite{Seco} & $E_{UB}$ \cite{LiebHertel} & $E_{ub}$ \cite{Seco} \\
\hline
$40$ & $-4433.13$ & $-3811.26$ & $-3906.96$  & $-3652.33$ & $-3534.86$ & $-3538.99$ \\
$41$ & $-4698.23$ & $-4042.55$ & - & $-3875.68$ & - & $-3753.55$ \\
$42$ & $-4972.44$ & $-4281.75$ & - & $-4102.74$ & - & $-3975.44$ \\
$43$ & $-5255.88$ & $-4528.96$ & - & $-4337.42$ & - & $-4204.78$ \\
$44$ & $-5548.65$ & $-4784.26$ & - & $-4579.82$ & - & $-4441.87$ \\
$45$ & $-5850.85$ & $-5047.75$ & - & $-4482.20$ & - & $-4685.80$ \\
$46$ & $-6162.61$ & $-5319.51$ & - & $-5079.65$ & - & $-4937.78$ \\
$47$ & $-6484.02$ & $-5599.63$ & - & $-5345.25$ & - & $-5197.52$ \\
$48$ & $-6815.16$ & $-5888.20$ & -  & $-5618.94$  & - & $-5465.11$ \\
$49$ & $-7156.11$ & $-6185.31$ & -  & $-5897.28$ & - & $-5740.16$ \\
$50$ & $-7506.98$ & $-6491.05$ & -  & $-6185.38$ & - & $-6022.93$ \\
$51$ & $-7867.88$ & $-6805.50$ & -  & $-6480.44$ & - & $-6313.48$ \\
$52$ & $-8238.91$ & $-7128.75$ & -  & $-6835.21$ & - & $-6661.78$ \\
$53$ & $-8620.17$ & $-7460.90$ & -  & $-7095.25$ & - & $-6917.98$ \\
$54$ & $-9011.78$ & $-7802.02$ & -  & $-7414.48$ & - & $-7232.13$ \\
$55$ & $-9413.82$ & $-8152.21$ & -  & $-7742.07$ & - & $-7553.93$ \\
$56$ & $-9826.41$ & $-8511.56$ & -  & $-8072.62$ & - & $-7883.54$ \\
$57$ & $-10249.64$ & $-8880.14$ & -  & $-8419.08$ & - & $-8221.06$ \\
$58$ & $-10683.63$ & $-9258.06$ & -  & $-8773.26$ & - & $-8566.91$ \\
$59$ & $-11128.47$ & $-9645.40$ & -  & $-9136.60$ & - & $-8921.18$ \\
$60$ & $-11584.27$ & $-10042.25$ & $-10020.50$  & $-9509.28$& $-9132.60$ & $-9283.88$ \\
$61$ & $-12039.50$ & $-10448.11$ & -  & $-9882.34$ & - & $-9655.09$ \\
$62$ & $-12504.77$ & $-10863.05$ & -  & $-10270.32$ & - & $-10034.95$ \\
$63$ & $-12980.15$ & $-11287.13$ & -  & $-10667.60$ & - & $-10423.54$ \\
$64$ & $-13465.72$ & $-11720.40$ & -  & $-11075.70$ & - & $-10820.61$ \\
$65$ & $-13961.52$ & $-12162.93$ & -  & $-11492.38$ & - & $-11226.57$ \\
$66$ & $-14467.64$ & $-12614.76$ & -  & $-11918.28$ & - & $-11641.45$ \\
$67$ & $-14984.14$ & $-13075.95$ & -  & $-12358.37$ & - & $-12065.29$ \\
$68$ & $-15511.08$ & $-13546.57$ & -  & $-12804.31$ & - & $-12498.15$ \\
$69$ & $-16048.54$ & $-14026.67$ & -  & $-13261.14$ & - & $-12940.17$ \\
$70$ & $-16596.58$ & $-14516.31$ & -  & $-13719.54$ & - & $-13391.45$ \\
$71$ & $-17155.26$ & $-15015.54$ & -  & $-14193.98$ & - & $-13851.81$ \\
$72$ & $-17724.65$ & $-15524.42$ & -  & $-14666.20$ & - & $-14321.25$ \\
$73$ & $-18304.82$ & $-16043.01$ & -  & $-15153.26$ & - & $-14799.81$ \\
$74$ & $-18895.84$ & $-16571.37$ & -  & $-15644.81$ & - & $-15287.54$ \\
$75$ & $-19497.77$ & $-17109.55$ & -  & $-16142.51$ & - & $-15784.53$ \\
$76$ & $-20110.67$ & $-17657.61$ & -  & $-16654.23$ & - & $-16290.65$ \\
$77$ & $-20734.62$ & $-18215.60$ & -  & $-17178.24$ & - & $-16806.11$ \\
$78$ & $-21369.68$ & $-18783.59$ & -  & $-17705.83$ & - & $-17330.94$ \\
$79$ & $-22015.91$ & $-19361.64$ & -  & $-18249.85$ & - & $-17865.21$ \\
$80$ & $-22673.39$ & $-19949.79$ & $-19607.08$  & $-18812.97$ & $-18090.35$ & $-18408.99$ \\
$81$ & $-23342.18$ & $-20548.10$ & -  & $-19356.51$ & - & $-18961.83$ \\
$82$ & $-24022.33$ & $-21156.64$ & -  & $-19928.40$ & - & $-19524.01$ \\
$83$ & $-24713.93$ & $-21775.46$ & -  & $-20507.70$ & - & $-20095.58$ \\
$84$ & $-25417.04$ & $-22404.61$ & -  & $-21093.76$ & - & $-20676.50$ \\
$85$ & $-26131.71$ & $-23044.16$ & -  & $-21692.17$ & - & $-21266.88$ \\
$86$ & $-26858.02$ & $-23694.15$ & -  & $-22292.92$ & - & $-21866.77$ \\
\hline
\end{tabular}
\caption{Comparison of our lower bounds with the lower and upper bounds from \cite{LiebHertel} and \cite{Seco}.}
\end{center}
\end{figure}
\renewcommand*{\figurename}{Figure}

Our conclusion is that for given $N$ and $Z$ our method allows for a fast computation of a lower bound for $E(N,Z)$ that is completely rigorous in the sense that numerical computations can be entirely avoided. However, especially for large values of $N$, the numerical methods from \cite{LiebHertel} and \cite{Seco} seem to yield significantly better results, although our lower bounds also become better under more realistic assumptions on the initial values $N_0$, $Z_0$, $E_0$ and on the ionization bound $f$. In particular, our lower bounds then improve upon the lower bounds from  \cite{LiebHertel} and \cite{Seco} for small atomic numbers $N$. Nevertheless, as already mentioned in Section \ref{secmain}, the main advantage of our method compared to the methods from \cite{LiebHertel} and \cite{Seco} is that our lower bounds can be explicitly expressed in terms of a series expansion (see Eq. \eqref{Literation}, below), which serves as the starting point in the derivation of Theorem \ref{theo2}. We are convinced that Theorem \ref{theo2} will in turn allow for  improvements of other results that make use of such an energy estimate. 

\section{Proof of Theorems \ref{theo1}, \ref{theo3} and \ref{thmEderivative}} \label{sectheo13}
To keep the notation short, we first rewrite the Hamiltonian from Eq. \eqref{DefHamiltonian} as
\begin{align}
H_N(Z)=T(N) - Z V(N) + I(N)\,, \label{H_N(Z)decomp}
\end{align}
where the operators
\begin{align}
T(N):= \sum_{i=1}^{N} - \dfrac{1}{2} \Delta_{x_i}, \qquad \qquad
V(N):=\sum_{i=1}^N \frac{1}{\left|x_i\right|}, \qquad \qquad
I(N)&:= \sum_{1 \leq i< j \leq N} \frac{1}{\left|x_i-x_j\right|}
\label{DefTVW}
\end{align}
describe the kinetic energy, the potential energy and the interaction energy, respectively.

\begin{lemma} \label{lm01}
For fixed $N \in \N$, the function $E(N, \,\cdot\,): (0,\infty) \rightarrow \R, Z \mapsto E(N,Z)$ is concave. In particular, it is locally Lipschitz (hence continuous) and differentiable in all but perhaps countably many points. Moreover, the one-sided derivatives 
\begin{align*}
\dfrac{\partial E}{\partial Z_{+}}(N,Z)= \lim_{Z' \downarrow Z} \dfrac{E(N,Z')-E(N,Z)}{Z'-Z},
\qquad \qquad  \dfrac{\partial E}{\partial Z_{-}}(N,Z)= \lim_{Z' \uparrow Z} \dfrac{E(N,Z')-E(N,Z)}{Z'-Z}
\end{align*}
exist for all $Z>0$ and define monotonically decreasing functions with respect to $Z$.   
\end{lemma}

\noindent\textit{Remark.} Similarly to \cite[Theorem 4.3.18 (Part I)]{Thirring} one can show that even $-\sqrt{-E(N,\,\cdot\,)}$ defines a concave function (see also \cite[Theorem 4.5.4 (Part I)]{Thirring}). But since we do not need this stronger statement in the following, we omit the proof. 

\begin{proof}
Although the proof that $E(N, \,\cdot\,)$ defines a concave function is standard, we briefly recall the argument.  
Since $H_N(Z)$ is linear in $Z$, we have for all $\lambda \in [0,1]$ and $Z, Z' >0$ the identity
\begin{align*}
H_N((1-\lambda)Z + \lambda Z')= (1-\lambda) H_N(Z)+ \lambda H_N(Z') \,.
\end{align*} 
The concavity of $E(N, \,\cdot\,)$ is now a direct consequence of the variational principle (cf. Eq. \eqref{VarPrinciple}):
\begin{align}
E(N, (1-\lambda)Z + \lambda Z') &= \inf_{\psi \in \mathcal{H}^2, \|\psi\|=1} \Braket{\psi| H_N((1-\lambda)Z + \lambda Z') \psi} 
\nonumber \\ & \geq (1-\lambda)  \inf_{\psi \in \mathcal{H}^2, \|\psi\|=1} \Braket{\psi| H_N(Z) \psi} + \lambda  \inf_{\psi \in \mathcal{H}^2, \|\psi\|=1} \Braket{\psi| H_N(Z') \psi} \nonumber \\
&= (1-\lambda) E(N,Z) + \lambda E(N,Z')\,. \nonumber 
\end{align}
The other properties  of $E(N, \,\cdot\,)$ are immediate consequences of its concavity (note that $-E(N, \,\cdot\,)$ is  then a convex function and for convex functions the desired properties are well-known, see e.g. \cite{ConvexAnalysis}). 
\end{proof}

\begin{lemma} \label{lm02}
Let $N \in \N$, $0<Z_0<Z<Z_1$ and suppose that  
$E(N,Z)$ is an eigenvalue of $H_N(Z)$ (for this each of the conditions $Z>N-1$ and $E(N,Z)<E(N-1,Z)$ is sufficient). Then, for all $\psi \in \text{GS}_N(Z)$ (cf. Eq. \eqref{DefGS}),
\begin{align}
\dfrac{E(N,Z)-E(N,Z_1)}{Z-Z_1} \leq  -\Braket{\psi 
 | V(N) \psi } \leq \dfrac{E(N,Z)-E(N,Z_0)}{Z-Z_0}\,. \label{DiffQuotEst}
\end{align}
Moreover, if the classical partial derivative $\dfrac{\partial E}{\partial Z}(N,Z)$ exists, then for all $\psi \in \text{GS}_N(Z)$,
\begin{align}
\dfrac{\partial E}{\partial Z}(N,Z) = -\Braket{\psi 
 |V(N) \psi }\,. \label{DiffE}
\end{align}
\end{lemma}

\begin{proof}
The fact that each of the conditions  $Z>N-1$ and $E(N,Z)<E(N-1,Z)$ implies that $E(N,Z)$ is an eigenvalue of $H_N(Z)$ follows from Eq. \eqref{HVZ} and \cite{Zhislin1, Zhislin2}. Now, suppose that  $E(N,Z)$ is an eigenvalue of $H_N(Z)$ and let
 $\psi \in \text{GS}_N(Z)$ , i.e. $\psi \in \mathcal{H}^2$, $H_N(Z)\psi = E(N,Z)\psi$ and $\|\psi\|=1$. Then it follows from the variational principle (cf. Eq. \eqref{VarPrinciple}) that for all $Z' > 0$,
\begin{align*}
E(N,Z') &\leq \Braket{\psi| H_N(Z') \psi}  \nonumber \\
& = \Braket{\psi| H_N(Z) \psi} + (Z-Z') \Braket{\psi  | V(N) \psi }\nonumber \\
&= E(N,Z) + (Z-Z') \Braket{\psi  |V(N)\psi }\,. \nonumber
\end{align*}  
The two inequalities in Estimate \eqref{DiffQuotEst} now follow by choosing $Z'=Z_0 < Z$ and $Z'=Z_1 > Z$, respectively. Finally, if the classical partial derivative $\frac{\partial E}{\partial Z}(N,Z)$ exists, then the left- and the right-hand side in Estimate \eqref{DiffQuotEst} both converge to $\frac{\partial E}{\partial Z}(N,Z)$ as $Z_1 \downarrow Z$ and $Z_0 \uparrow Z$, respectively. This implies the second part of the lemma.
\end{proof}

The following theorem is well-known (see e.g. \cite[Theorem 4.5.5 (Part I)]{Thirring}), but for the convenience of the reader we recall its proof:

\begin{proposition}[Virial theorem] \label{prop3}
Let $N \in \N$, $Z>0$ and $\psi \in \mathcal{H}^2$ with $\|\psi\|=1$. Then $H_N(Z)\psi= E \psi$ for some $E \in \R$ implies that
\begin{align*}
E = -  \Braket{\psi | T(N) \psi} \,.
\end{align*}
In particular, $H_N(Z)$ has no eigenvalues $E \geq 0$.
\end{proposition}

\begin{proof}
For $\lambda>0$, we consider the unitary scaling $U_\lambda \in \LL(\mathcal{H})$ defined by
\begin{align}
(U_\lambda \psi)(x_1,s_1,x_2,s_2,\ldots, x_N, s_N):= \lambda^{3N/2} \psi(\lambda x_1,s_1,\lambda x_2,s_2,\ldots, \lambda x_N, s_N)\,. \label{DefU_Z}
\end{align}
Then a straightforward computation shows that
\begin{align*}
(U_{\lambda})^{*} H_N(Z) U_{\lambda} = \lambda^2 T(N) - Z \lambda V(N) + \lambda I(N). 
\end{align*}
Now, on the one hand $H_N(Z)\psi= E \psi$ implies that $\Braket{(U_\lambda)^{*} \psi | \left(T(N) - Z V(N) + I(N) - E \right)\psi }=0$, and  on the other hand
\begin{align*}
0 &= \Braket{(H_N(Z)-E) \psi | U_\lambda \psi  } \\
 &= \Braket{ \psi | (H_N(Z)-E) U_\lambda \psi  }\\
&= \Braket{ (U_\lambda)^{*} \psi | (U_\lambda)^{*} (H_N(Z)-E) U_\lambda \psi  }\\
&= \Braket{ (U_\lambda)^{*} \psi | \left(\lambda^2 T(N) - Z \lambda V(N) + \lambda I(N) - E \right)\psi }\,.
\end{align*}
Combining these two equations leads to
\begin{align*}
\Braket{ (U_\lambda)^{*} \psi | \left((\lambda^2-1)\,T(N)- Z(\lambda-1)\, V(N) + (\lambda-1)\, I(N) \right)\psi }=0\qquad \quad (\lambda>0).
\end{align*}
Since $(U_\lambda)^{*} \psi \rightarrow \psi$ strongly as $\lambda \uparrow 1$, after division  by $(\lambda-1)$ we find that in the limit $\lambda \uparrow 1$,
\begin{align*}
0 = \Braket{ \psi | \left(2T(N)- Z V(N) + I(N) \right)\psi  }
= E + \Braket{\psi | T(N) \psi}\,.
\end{align*}
In particular, $H_N(Z)$ has no eigenvalues $E \geq 0$ because $\Braket{\psi | T(N) \psi}>0$.
\end{proof}

\begin{proposition}\label{prop4}
Let $N \in \N$. Then, for all $\psi \in \mathcal{H}^2$ with $\|\psi\|=1$,
\begin{align*}
\Braket{\psi | V(N) \psi} \leq 2\sqrt{\sigma_N} \sqrt{ \Braket{\psi | T(N) \psi}}\,.
\end{align*}
\end{proposition}

\begin{proof}
First, we recall from Eqs. \eqref{DefSigmaN} and \eqref{Defh_N} that $-\sigma_N$ is the ground-state energy of $h_N=T(N)-V(N)$.
With the unitary scaling $U_\lambda \in \LL(\mathcal{H})$ defined in Eq. \eqref{DefU_Z}, it follows that for all $\psi \in \mathcal{H}^2$ with $\|\psi\|=1$,
\begin{align*}
-\sigma_N \leq \Braket{\psi| (U_{\lambda})^{*} h_N U_{\lambda} \psi}
= \lambda^2 \Braket{\psi | T(N) \psi} - \lambda \Braket{\psi | V(N) \psi} \qquad \qquad (\lambda>0)\,,
\end{align*}
which is equivalent to
\begin{align*}
\Braket{\psi | V(N) \psi} \leq  \lambda \Braket{\psi | T(N) \psi} + \dfrac{\sigma_N}{\lambda} \qquad \qquad (\lambda>0)\,.
\end{align*}
The desired estimate follows by optimizing in $\lambda$ (the optimal value is $\lambda=\sqrt{\sigma_N / \Braket{\psi | T(N) \psi}}$).
\end{proof}

\begin{proposition}\label{prop5}
For all $N \in \N$ and $Z_0,Z_1>0$ with $Z_0 \leq Z_1$,
\begin{align*}
\sqrt{\left| E(N,Z_1) \right|} \leq \sqrt{\left| E(N,Z_0) \right|} 
+ \sqrt{\sigma_N} (Z_1-Z_0)\,.
\end{align*}
\end{proposition}

\begin{proof}
We shall prove below that the left-hand derivative $\frac{\partial E}{\partial Z_{-}}(N,Z)$ defined in Lemma \ref{lm01} satisfies the differential inequality
\begin{align}
\frac{\partial E}{\partial Z_{-}}(N,Z) \geq -2 \sqrt{\sigma_{N}} \sqrt{\left| E(N,Z) \right|}. \label{E'est}
\end{align}
In particular, the same differential inequality then holds with $\frac{\partial E}{\partial Z_{-}}(N,Z)$ replaced by $\frac{\partial E}{\partial Z}(N,Z)$, provided that $Z>0$ is chosen so that the classical partial derivative $\frac{\partial E}{\partial Z}(N,Z)$ exists. By Lemma \ref{lm01},  $E(N, \, \cdot\,)$ is locally Lipschitz and hence absolutely continuous. The same is also true for $\sqrt{\left|E(N, \, \cdot\,)\right|}$. 
Therefore, the fundamental theorem of calculus holds for $\sqrt{\left|E(N, \, \cdot\,)\right|}$, which in combination with 
Estimate \eqref{E'est} allows us to derive the desired estimate
\begin{align*}
\sqrt{\left| E(N,Z_1) \right|} - \sqrt{\left| E(N,Z_{0}) \right|} = -\int\limits_{Z_0}^{Z_1} 
\dfrac{\dfrac{\partial E}{\partial Z}(N,z)}{2 \sqrt{\left| E(N,z) \right|}}\, dz \leq \int\limits_{Z_0}^{Z_1} \sqrt{\sigma_{N}} \, dz = \sqrt{\sigma_{N}} (Z_1-Z_{0})\,.
\end{align*}

It remains to establish Estimate \eqref{E'est} for arbitrary $Z>0$. For this purpose, we use induction over $N \in \N$. For $N=1$, Estimate \eqref{E'est} follows from the explicit formula $E(1,Z)=-Z^2/2$ and $\sigma_1=1/2$. 
For the induction step, we now suppose that for some $N \geq 2$, 
\begin{align}
\frac{\partial E}{\partial Z_{-}}(N-1,Z) \geq -2 \sqrt{\sigma_{N-1}} \sqrt{\left| E(N-1,Z) \right|} \qquad \qquad (Z>0)\,, \label{indV}
\end{align}
so we have to prove that Estimate \eqref{E'est} holds for all $Z>0$. 
As a first step, we prove Estimate \eqref{E'est} in two special cases. The first case is that $Z>0$ is chosen so that $E(N,Z)<E(N-1,Z)$. Then, it follows from Lemma \ref{lm02} that for any $\psi \in \text{GS}_N(Z)$,
\begin{align*}
\frac{\partial E}{\partial Z_{-}}(N,Z) \geq  -\Braket{\psi 
 | V(N) \psi }\,.
\end{align*}
Hence, using Propositions \ref{prop3} and \ref{prop4} to estimate the right-hand side, we obtain Estimate \eqref{E'est} for this particular choice of $Z>0$. The second case is that $Z>0$ is chosen so that $E(N,Z')=E(N-1,Z')$ holds for all $Z'$ in some open neighborhood of $Z$. But then it follows from the induction hypothesis \eqref{indV} together with the estimates $\sigma_{N-1} \leq \sigma_N$ and $E(N,Z)  =  E(N-1,Z) < 0$ that
\begin{align*}
\frac{\partial E}{\partial Z_{-}}(N,Z) = \frac{\partial E}{\partial Z_{-}}(N-1,Z) &\geq -2 \sqrt{\sigma_{N-1}} \sqrt{\left| E(N-1,Z) \right|} \\ 
& \geq -2 \sqrt{\sigma_{N}} \sqrt{\left| E(N,Z) \right|}\,.
\end{align*}
This means that we have now established Estimate \eqref{E'est} for two particular choices of $Z>0$, which cover some dense subset of $(0,\infty)$. Note that this implies that for \textit{any} fixed $Z>0$ there exists a sequence $(Z_n)_{n \in \N}$ so that Estimate \eqref{E'est} holds for each $Z_n$ and $Z_n \downarrow Z$ as $n \to \infty$. Hence, using that $\frac{\partial E}{\partial Z_{-}}(N,\,\cdot\,)$ is monotonically decreasing by Lemma \ref{lm01}, we find that
\begin{align*}
\frac{\partial E}{\partial Z_{-}}(N,Z) \geq \frac{\partial E}{\partial Z_{-}}(N,Z_n) \geq -2 \sqrt{\sigma_{N}} \sqrt{\left| E(N,Z_n) \right|} \qquad \qquad (n \in \N)\,.
\end{align*}
Estimate \eqref{E'est} for the given value of $Z>0$ now emerges in the limit $n \to \infty$ since $E(N,\,\cdot\,)$ is a continuous function.
\end{proof}

\noindent\textit{Proof of Theorem \ref{theo1}.}
First, we note that the desired estimate  $E(N,Z) \geq  -L(N,Z)^2$ is equivalent to 
\begin{align} \label{Toprove}
 \sqrt{\left| E(N,Z) \right|} \leq L(N,Z) \qquad \qquad (N \geq N_0, Z \geq Z_0)\,.
\end{align} 
To prove Estimate \eqref{Toprove}, we use induction over $N $. For $N=N_0$, Estimate \eqref{Toprove} follows from the Definition \eqref{DefL0} of $L(N_0,Z)$, Proposition \ref{prop5} and the assumption that    $E_0 \in (-\infty, E(N_0,Z_0)]$. Suppose now that for some $N \in \N$, $N \geq N_0+1$,
\begin{align} \label{IndHyp}
\sqrt{\left| E(N-1,Z) \right|} \leq L(N-1,Z) \qquad \qquad (Z \geq Z_0)\,.
\end{align} 
To prove that Estimate \eqref{Toprove} then also holds for $N$, we distinguish between two cases. The first case is that $Z_0 \leq Z \leq f^{-1}(N)$ (note that this case is not void because the assumption that $Z_0 \in (0,f^{-1}(N_0+1)]$ and the strict monotonicity of $f$ imply that $f^{-1}(N) \geq Z_0$ for all $N \geq N_0+1$). 
Now, Estimate \eqref{Ncbound} and the strict monotonicity of $f$ imply that $N>N_c(Z)$, and hence $E(N,Z)=E(N-1,Z)$. Therefore, we arrive at
\begin{align}
\sqrt{\left| E(N,Z) \right|} = \sqrt{\left| E(N-1,Z) \right|} \leq L(N-1,Z) = L(N,Z)\,,\qquad \quad (Z_0 \leq Z \leq f^{-1}(N)),  \label{firstcase}
\end{align}
where the induction hypothesis \eqref{IndHyp} was used for the inequality.

In the second case, where $Z > f^{-1}(N)$, Proposition \ref{prop5} and Estimate \eqref{firstcase} with $Z=f^{-1}(N)$ imply that
\begin{align*}
\sqrt{\left| E(N,Z) \right|} &\leq \sqrt{\left| E(N,f^{-1}(N)) \right|} + \sqrt{\sigma_N} (Z-f^{-1}(N)) \\
&\leq L(N-1,f^{-1}(N)) + \sqrt{\sigma_N} (Z-f^{-1}(N)) \\
&= L(N,Z)\,.
\end{align*}
This completes the proof of Estimate \eqref{Toprove}. \hfill$\Box$

\bigskip

\noindent\textit{Proof of Theorem \ref{theo3}.}
The fact that $E(N, \,\cdot\,): (0,\infty) \rightarrow \R, Z \mapsto E(N,Z)$ is concave and hence differentiable in all but perhaps countably many points has already been verified in Lemma \ref{lm01}. To show that $E(N, \,\cdot\,)$ is strictly monotonically decreasing, we have to prove that $Z'>Z$ implies that $E(N,Z')<E(N,Z)$.
In view of the variational principle (cf. Eq. \eqref{VarPrinciple}), there exists $\psi \in \mathcal{H}^2$ with $\|\psi\|=1$ and 
\begin{align*}
\braket{\psi| H_N(Z)\psi} < \dfrac{Z}{Z'}  E(N,Z)\,. 
\end{align*}
But then, using again the variational principle, we find that
\begin{align*}
E(N,Z') &\leq \braket{\psi| H_N(Z')\psi} = \dfrac{Z'}{Z} \braket{\psi| H_N(Z)\psi} - \left( \dfrac{Z'}{Z} - 1 \right) \Braket{\psi| (T(N)+ I(N)) \psi}  < E(N,Z)\,,
\end{align*}
where we also used that $T(N)$ and $I(N)$ are positive operators. 
It remains to show that if $E(N,Z)$ is an eigenvalue of $H_N(Z)$ and $E(N, \,\cdot\,)$ is differentiable at this particular value of $Z$, then for all $\psi \in \text{GS}_N(Z)$,
\begin{align}
\Braket{\psi|T(N)\psi} &= \left|E(N,Z)\right|\,, \label{TId} \\
\Braket{\psi|V(N)\psi} &= -\dfrac{\partial E}{\partial Z}(N,Z)\,, \label{VId}\\
\Braket{\psi|I(N)\psi} &= 2 E(N,Z) - Z\dfrac{\partial E}{\partial Z}(N,Z) \,. \label{IId}
\end{align}
Identities \eqref{TId} and \eqref{VId} for the expectation values of 
the kinetic energy and the potential energy immediately follow from the Virial theorem (Proposition \ref{prop3}) and Lemma \ref{lm02}, respectively. Identity \eqref{IId} for the expectation value of the interaction energy then follows by inserting these identities in 
\begin{align*}
E(N,Z)=\Braket{\psi | H_N(Z)\psi}=  \Braket{\psi|T(N)\psi} - Z\Braket{\psi | V(N)\psi}+\Braket{\psi|I(N)\psi}\,.
\end{align*} \hfill$\Box$

\bigskip

\noindent\textit{Proof of Theorem \ref{thmEderivative}.}
From Lemma \ref{lm02} it immediately follows that
\begin{align}
\frac{\partial E}{\partial Z_{+}}(N,Z)&\leq \min_{\psi \in \text{GS}_N(Z)} - \Braket{\psi|V(N)\psi}\,, \label{righthandupper} \\
\frac{\partial E}{\partial Z_{-}}(N,Z) &\geq \max_{\psi \in \text{GS}_N(Z)} - \Braket{\psi|V(N)\psi}\,,
\end{align}
where $\text{GS}_N(Z)$ denotes the set of all normalized ground states (cf. Eq. \eqref{DefGS}). 
For the reverse inequalities, we only consider the right-hand derivative because the proof for the left-hand derivative is very similar. Since the ground-state energy $E(N,Z)$ is continuous w.r.t $Z$ and 
$E(N,Z)<E(N-1,Z)$
by assumption, there exists a sequence $(Z_n)_{n \in \N}$ with $Z_n \downarrow Z$ as $n \to \infty$ and $E(N,Z_n)<E(N-1,Z_n)$ for all $n \in \N$. In view of the HVZ theorem (cf. Eq. \eqref{HVZ}), this implies that 
$E(N,Z)$ and $E(N,Z_n)$ are discrete eigenvalues of $H_N(Z)$, respectively, $H_N(Z_n)$ that lie  below the bottom of the essential spectrum. For later purposes, we note that this also implies that there exists $\delta>0$ such that
\begin{align}
 \sigma(H_N(Z)) \cap (E(N,Z)-\delta, E(N,Z)+\delta)=\left\{E(N,Z) \right\}\,. \label{spectralgap} 
\end{align}
Moreover, $\text{GS}_N(Z_n)$ is non-empty for every $n \in \N$, i.e. there exists $\psi_n \in \mathcal{H}^2$ with $H_N(Z_n)\psi_n=E(N,Z_n)\psi_n$ and $\|\psi_n\|=1$.
From the Virial theorem (cf. Proposition \ref{prop3}), it now follows that $\| \nabla \psi_n\|^2 = 2 \braket{\psi_n|T(N)\psi_n}=2\left|E(N,Z_n)\right|$ is also uniformly bounded in $n \in \N$, so the sequence $(\psi_n)_{n \in \N}$ is uniformly bounded w.r.t. the norm of $\mathcal{H}^1$. Hence, the Banach-Alaoglu theorem guarantees the existence  of a weakly convergent subsequence, which for simplicity is also denoted by $(\psi_n)_{n \in \N}$, i.e. $\psi_n \rightharpoonup \psi_0$ in $\mathcal{H}^1$ as $n \to \infty$ for some $\psi_0 \in \mathcal{H}^1$. 
Our goal is to show that the weak limit $\psi_0$ is in fact a normalized ground state of $H_N(Z)$, i.e. $\psi_0 \in \text{GS}_N(Z)$, and that 
$\Braket{\psi_n|V(N)\psi_n} \rightarrow \Braket{\psi_0|V(N)\psi_0}$ as $n \to \infty$. First, we note that Hardy's inequality implies that $V(N)$ is relatively bounded w.r.t. the $\mathcal{H}^1$-norm, i.e. $V(N): \mathcal{H}^1 \rightarrow \mathcal{H}$ is a bounded operator. Hence, the weak convergence $\psi_n \rightharpoonup \psi_0$ in $\mathcal{H}^1$  implies weak convergence $V(N)\psi_n \rightharpoonup V(N)\psi_0$ in $\mathcal{H}$. However, to obtain convergence of $\Braket{\psi_n|V(N)\psi_n}$ towards $\Braket{\psi_0|V(N)\psi_0}$, we need strong convergence $\psi_n \to \psi_0$  in $\mathcal{H}$. For  this purpose, we note that the weak convergence of $(V(N)\psi_n)_{n \in \N}$ in $\mathcal{H}$ implies that $\|V(N)\psi_n\|$ is uniformly bounded in $n \in \N$ and hence
\begin{align*}
\| (H_N(Z) - E(N,Z)) \psi_n \| &\leq     
\| (H_N(Z) - E(N,Z_n)) \psi_n \| + |E(N,Z_n)-E(N,Z)| \nonumber \\
&= \| (H_N(Z) - H_N(Z_n)) \psi_n \| + |E(N,Z_n)-E(N,Z)| \nonumber \\
&= |Z-Z_n| \|V(N)\psi_n\| + |E(N,Z_n)-E(N,Z)| \rightarrow 0 \qquad (n \to \infty)\,. 
\end{align*}
Let now $P$ denote the orthogonal projection onto the finite-dimensional ground-state subspace
\begin{align*}
\left\{\psi \in \mathcal{H}^2 \; : \; H_N(Z)\psi=E(N,Z)\psi \right\}\,.  
\end{align*}
Obviously, $(H_N(Z)-E(N,Z))P\psi_n=0$, so with $P^{\perp}:=1-P$ we obtain that also $\| (H_N(Z) - E(N,Z))P^{\perp}\psi_n \| \to 0$ as $n \to \infty$.
Using this together with Eq. \eqref{spectralgap} and the spectral theorem, we conclude that
\begin{align*}
\delta \|P^{\perp}\psi_n\| \leq \| (H_N(Z) - E(N,Z))P^{\perp}\psi_n \|   \rightarrow 0 \qquad \quad (n \to \infty)\,.   
\end{align*}
On the one hand, we thus obtain that $P^{\perp}\psi_n \to 0$ as $n \to \infty$. On the other hand, since $P$ is compact as a finite-rank operator, it improves the weak convergence $\psi_n \rightharpoonup \psi_0$ to strong convergence $P\psi_n \to P\psi_0$ as $n \to \infty$. 
Altogether, we thus obtain that $\psi_n \to P\psi_0$ strongly as $n \to \infty$. But since strong convergence implies weak convergence and $\psi_n$ also converges weakly to  $\psi_0$, this implies that $P\psi_0=\psi_0$ because weak limits are unique, and therefore $\psi_n \to \psi_0$ strongly as $n \to \infty$. In particular, it follows  that $\|\psi_0\|= \lim_{n \to \infty} \|\psi_n\|=1$ and thus $\psi_0 \in \text{GS}_N(Z)$ because $P\psi_0=\psi_0$. Furthermore, as already explained above, this also implies that
\begin{align*}
\Braket{\psi_n|V(N)\psi_n} \to \Braket{\psi_0|V(N)\psi_0} \qquad \quad (n \to \infty)\,.  
\end{align*}
Finally, we are now in the position to derive the reverse of Estimate \eqref{righthandupper}:
\begin{align*}
\dfrac{E(N,Z_n)-E(N,Z)}{Z_n-Z} \geq  
\dfrac{\Braket{\psi_n|H_N(Z_n)\psi_n}-\Braket{\psi_n|H_N(Z)\psi_n}}{Z_n-Z} 
= -\Braket{\psi_n| V(N) \psi_n}\,,
\end{align*}
which in the limit $n \to \infty$ results in the desired lower bound
\begin{align*}
\frac{\partial E}{\partial Z_{+}}(N,Z)\geq -\Braket{\psi_0| V(N) \psi_0} \geq  \min_{\psi \in \text{GS}_N(Z)} - \Braket{\psi|V(N)\psi}\,.
\end{align*}
\hfill$\Box$

\section{Proof of Theorem \ref{theo2}} \label{sectheo2}
The goal of this section is to prove Theorem \ref{theo2}, which is dedicated to (Case FER) with $q=2$.
As we will see below, the proof of Estimate \eqref{Elower} with $\beta=1.1185$, $\eta \in [1/\beta,\infty)$  and $Z \in (0,\eta N]$  can be reduced to the endpoint case
$Z=\eta N$, so we only have to show that for all $N \in \N$ and $\eta \in [1/\beta, \infty)$, 
\begin{align}
E(N,\eta N) \geq - \alpha(\eta, \beta)^2 \, \eta^2  N^{7/3} \,, \label{Elowergoal2}
\end{align}
where $\alpha(\eta, \beta)$ is defined by Eq. \eqref{Defalpha}.
To this end,  we are going to apply Theorem \ref{theo1} with the same parameters as in Eq. \eqref{firstchoice}, i.e. $N_0=2$, $Z_0=1$ and $E_0=-0.5774$. 
Recall from Eq. \eqref{fchoice} that in (Case FER) with $q=2$, Estimate \eqref{Ncbound} is satisfied for the function $f: (0,\infty) \rightarrow [1,\infty)$ defined by
\begin{align}
f(Z):=  \begin{cases}
2Z+1 & \qquad \textup{if} \; Z \leq 5.5\,, \\
1.22Z +3 Z^{1/3} & \qquad \textup{if} \; 5.5 < Z < Z_1\,, \\
\beta Z +k Z^{1/3} & \qquad \textup{if} \; Z \geq Z_1\,,
\end{cases} \label{currentf}
\end{align}
where $k:=4$ and $Z_1:=30.9$. 
We are going to compute 
$E_L(N_1,Z_1)=-L(N_1,Z_1)^2$ for $N_1:=\lfloor f(Z_1) \rfloor=47$.
Using the recursive Definition \eqref{DefL} of $L$ starting at $N_0=2$, $Z_0=1$ and $E_0=-0.5774$, a straightforward but lengthy computation yields that  $E_L(N_1,Z_1)\geq -2501$.
Using again Definition \eqref{DefL} together with the fact that $f^{-1}(N) > Z_1$ for all $N \geq N_1+1$, this implies that $E_L(N,Z_1)= E_L(N_1,Z_1) \geq  -2501$ for all $N \geq N_1$.
Next, we note that Definition \eqref{currentf} and the assumption that $\eta \in [1/\beta, \infty)$ imply that $f(\eta N) > N$ and hence $\eta N > f^{-1}(N)$ for all $N \in \N$. 
By iterating the recursive Definition \eqref{DefL} of $L$, we thus arrive at 
\begin{align}
L(N,\eta N) = L(N_1,Z_1) 
+ \sum_{n=N_1}^{N-1} l_n \sqrt{\sigma_n} +  \left( \eta N - f^{-1}(N) \right) \sqrt{\sigma_N} \qquad \qquad (N \geq N_1+1), \label{Literation}
\end{align}  
where $L(N_1,Z_1)= \sqrt{\left|E_L(N_1,Z_1)\right|} \leq 50.01$, $l_{N_1}:=f^{-1}(N_1+1)-Z_1$, and
\begin{align*}
l_n := f^{-1}(n+1) - f^{-1}(n)  \qquad \qquad (n \geq N_1+1)\,.
\end{align*}
To prove Estimate \eqref{Elowergoal2}, we are going to bound all contributions to the right-hand side of Eq. \eqref{Literation}. We start by deriving an upper bound for the length $l_n$ of the intervals, where the number of bound electrons remains constant.

\begin{lemma} \label{lm1}
For $N_1=47$ and $Z_1=30.9$, the sequence $(l_n)_{n \geq N_1}$ is strictly monotonically increasing with $0 < l_n < 1/\beta$ for all $n \geq N_1$.
\end{lemma}

\begin{proof}
The inequality $l_{47} < l_{48}$ can be verified by direct computation. For $n \geq 48$, the fundamental theorem of calculus yields that 
\begin{align*}
l_n = f^{-1}(n+1) - f^{-1}(n) = \int\limits_n^{n+1} (f^{-1})'(x) \, dx =  \int\limits_n^{n+1} \dfrac{1}{f'(f^{-1}(x))} \, dx\,.
\end{align*}
Using now that $f'(x)=\beta + k x^{-2/3}/3>\beta$ for all $x>Z_1$, we conclude that 
\begin{align*}
0< l_n =  \int\limits_n^{n+1} \dfrac{1}{f'(f^{-1}(x))} \, dx\,
< \dfrac{1}{\beta}\,.
\end{align*}
It remains to show that 
 $(l_n)_{n \geq 48}$ is strictly monotonically increasing. For this purpose, we first note that $f$, and hence $f^{-1}$, is strictly monotonically increasing, whereas $f'$ is strictly monotonically decreasing for $x > Z_1$. This implies that $f'(f^{-1}(x)) > f'(f^{-1}(x+1))$
 for $x \geq 48$, and hence we obtain for all $n \geq 48$ that
 \begin{align*}
 l_n=  \int\limits_n^{n+1} \dfrac{1}{f'(f^{-1}(x))} \, dx
 <  \int\limits_{n}^{n+1} \dfrac{1}{f'(f^{-1}(x+1))} \, dx
 = \int\limits_{n+1}^{n+2} \dfrac{1}{f'(f^{-1}(x))} \, dx = l_{n+1}\,.
 \end{align*}
\end{proof}
Using Lemma \ref{lm1} to estimate the right-hand side of Eq. \eqref{Literation}, we arrive at
\begin{align}
L(N,\eta N) &\leq L(N_1,Z_1) + \frac{1}{\beta} \sum_{n=N_1}^{N-1} \sqrt{\sigma_n} + \left( \eta N - f^{-1}(N) \right) \sqrt{\sigma_N} \nonumber \\
&=  L(N_1,Z_1) + \frac{1}{\beta} \sum_{n=N_1}^{N} \sqrt{\sigma_n} + \left( \eta N - \frac{1}{\beta} - f^{-1}(N) \right) \sqrt{\sigma_N}\,. \label{Ineq1}
\end{align} 
For a further estimate of the right-hand side, we need the following basic observation:

\begin{lemma}\label{lm2}
Let $n_0, n_1 \in \Z$ with $n_0 \leq n_1$ and let $g: [n_0, n_1] \rightarrow \R$ be a concave function. Then,
\begin{align*}
\sum\limits_{n=n_0}^{n_1} g(n) \leq \int\limits_{n_0}^{n_1} g(x) \,  dx + \dfrac{g(n_0) + g(n_1)}{2}\,.
\end{align*}
\end{lemma}

\begin{figure}[H] 
\begin{center}
\begin{tikzpicture}[scale=0.8, even odd rule, x=3.5cm, y=2cm, declare function={
      func(\x)={0.3+2.5*ln(\x)};
   }]
 
\draw[-{Stealth[length=10pt, width=8pt]}] (-0.5,0) -- (4.5,0) node[below left=3pt and 4pt]{\Large$x$};
\draw[-{Stealth[length=10pt, width=8pt]}] (0,-0.5) -- (0,4.2) node[below left= 3pt and 3pt]{\Large$y$};

   \draw[color=black] ({1},-.2) -- ({1},.0) node[below left =8pt and -9pt, color=black] {\Large $n_0$};
     \draw[color=black] ({4},-.2) -- ({4},.0) node[below left =8pt and -9pt, color=black] {\Large $n_1$};

\foreach \a in {1,2,3}  
\fill[red!40] plot[domain=(\a+1):\a] (\x,{0})%
                      -- plot[domain=\a:(\a+1)] (\x,{func(\a) +(func(\a+1)-func(\a))*(\x-\a)});   
 \foreach \a in {1,2,3}
 \draw[red] plot[domain=(\a+1):\a] (\x,{0})%
                      -- plot[domain=\a:(\a+1)] (\x,{func(\a) +(func(\a+1)-func(\a))*(\x-\a)});                       
\draw[red] (4,0)--(4,{func(4)});
\draw[blue] plot[domain=1:4] (\x,{func(\x)});

\node[color=blue] at (2.4,2.7) {\Large $g$};
\end{tikzpicture}
\caption{The red area is obtained by applying the composite trapezoidal rule with step width $1$ to the concave function $g: [n_0,n_1] \rightarrow \R$. This yields a lower bound for  $\int\limits_{n_0}^{n_1} g(x) \,  dx$.}
\label{Fig:IntegralConcave}
\end{center}
\end{figure}
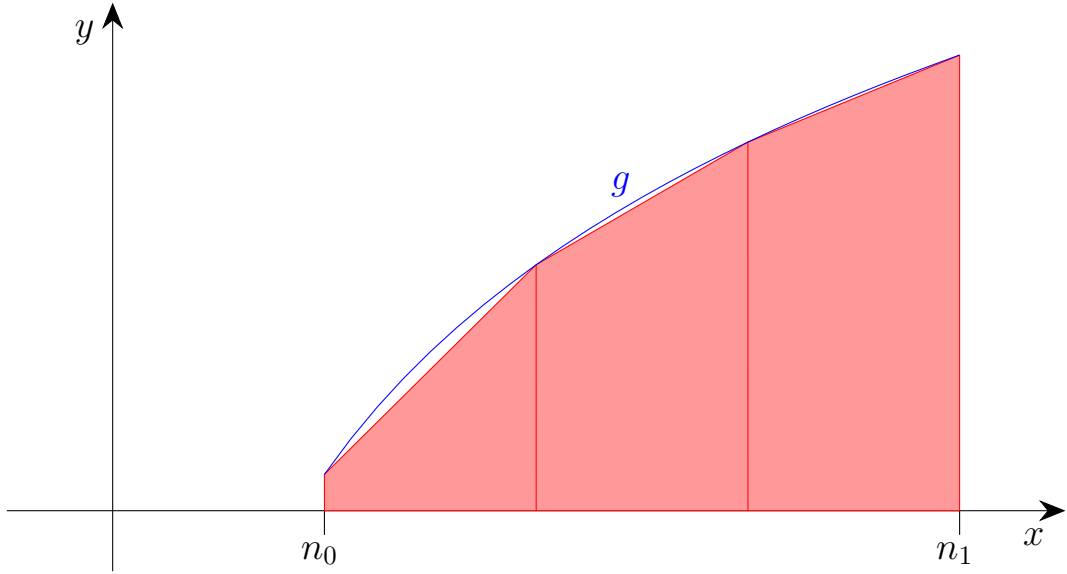

\begin{proof}
For $n_0=n_1$, both sides are equal to $g(n_0)$. For $n_0<n_1$,
the desired estimate essentially follows from the fact that for concave functions the composite trapezoidal rule yields a lower bound for the integral (as illustrated in Figure \ref{Fig:IntegralConcave}).
For the sake of completeness, we also give the explicit computation:
\begin{align*}
\int\limits_{n_0}^{n_1} g(x) \,  dx 
= \sum\limits_{n=n_0}^{n_1-1} \int\limits_{n}^{n+1} g(x) \,  dx 
&= \sum\limits_{n=n_0}^{n_1-1} \int\limits_{0}^{1} g((1-\lambda)n+ \lambda (n+1))  \,  d\lambda \nonumber \\
& \geq \sum\limits_{n=n_0}^{n_1-1} \int\limits_{0}^{1} (1-\lambda) g(n) + \lambda g(n+1)  \,  d\lambda \nonumber \\
&= \sum\limits_{n=n_0}^{n_1-1}  \dfrac{g(n+1) + g(n)}{2} 
= \sum\limits_{n=n_0}^{n_1} g(n) - \dfrac{g(n_0) + g(n_1)}{2}\,.
\end{align*}
\end{proof}

To bound the summands $\sqrt{\sigma_n}$ in Estimate \eqref{Ineq1}, we use Lemma \ref{lmA2} in the appendix, which implies that
\begin{align}
\sigma_n \leq 2^{-1} 3^{1/3} q^{2/3} n^{1/3} - \delta(N_1,q) \qquad \qquad (n \geq N_1), \label{sigmaest}
\end{align}
where $$\delta(N_1,q) = \frac{q}{4} - \frac{q^{4/3}}{3^{4/3} \cdot 8}  \left( 1+r(N_1)^{1/3}\right) N_1^{-1/3}$$
with $r(N_1)$ defined by Eq. \eqref{Defr} below.

\begin{lemma}\label{lm3}
Suppose that $N_1 \in \N$ and $\delta >0$ are chosen so that $\sigma_n \leq 2^{-1}3^{1/3} q^{2/3} n^{1/3} - \delta$ holds for all $n \geq N_1$. Then, for all $N \geq N_1+1$,
\begin{align*}
\sum\limits_{n=N_1}^{N} \sqrt{\sigma_n} \leq \; & \dfrac{2^{1/2}\cdot 3^{7/6}}{7}  q^{1/3} (N^{7/6}- N_1^{7/6} ) - \frac{2^{1/2} \cdot 3^{5/6}}{5} \delta q^{-1/3} 
(N^{5/6}- N_1^{5/6} ) \nonumber \\ & + \frac{3^{1/6}}{2^{3/2}} q^{1/3} (N^{1/6} + N_1^{1/6} ) - \frac{1}{2^{3/2} \cdot 3^{1/6}} \delta q^{-1/3} (N^{-1/6} + N_1^{-1/6} )\,. 
\end{align*}
\end{lemma}
\begin{proof}
The function $g: [N_1, N] \rightarrow \R$ with $g(x):=\sqrt{2^{-1}3^{1/3}q^{2/3}x^{1/3}-\delta}$ is concave because
$$g'(x)= \frac{q^{2/3}}{4} (3x)^{-2/3} \left(2^{-1}3^{1/3} q^{2/3}  x^{1/3}- \delta \right)^{-1/2} $$
is (strictly) monotonically decreasing. Hence, the assumption that $\sigma_n \leq g(n)^2$ for $n \geq N_1$ and Lemma \ref{lm2} yield the estimate
\begin{align}
\sum\limits_{n=N_1}^{N} \sqrt{\sigma_n} \leq \sum\limits_{n=N_1}^{N} g(n) \leq  \int\limits_{N_1}^{N} g(x) \,  dx + \dfrac{g(N_1) + g(N)}{2},  \label{est1}
\end{align}
where
\begin{align}
\int\limits_{N_1}^{N} g(x) \,  dx &= \int\limits_{N_1}^{N} \sqrt{2^{-1}3^{1/3}q^{2/3}x^{1/3}-\delta} \,  dx \leq 
\int\limits_{N_1}^{N} 2^{-1/2} 3^{1/6}q^{1/3} x^{1/6} - \frac{1}{2^{1/2} \cdot 3^{1/6}}  \delta q^{-1/3}  x^{-1/6} \,  dx \nonumber \\
&= \left[ \frac{2^{1/2} \cdot 3^{7/6} }{7} q^{1/3} x^{7/6} - \frac{2^{1/2} \cdot 3^{5/6}}{5} \delta q^{-1/3} x^{5/6} \right]_{x=N_1}^{x=N}\nonumber \\
&= \frac{2^{1/2} \cdot 3^{7/6}}{7} q^{1/3} (N^{7/6} - N_1^{7/6})
- \frac{2^{1/2} \cdot 3^{5/6}}{5}   \delta q^{-1/3} (N^{5/6} - N_1^{5/6}).
\label{est2}
\end{align}
Here, the inequality in the first line is a consequence of the elementary estimate
\begin{align}
\sqrt{a-\delta} \leq \sqrt{a} - \frac{\delta}{2 \sqrt{a}} \qquad \qquad (0 \leq \delta \leq a)\,. \label{elementaryest}
\end{align}
Furthermore, Estimate \eqref{elementaryest} also implies that
\begin{align}
 \dfrac{g(N_1) + g(N)}{2} &= \frac{1}{2} \left( \sqrt{2^{-1}3^{1/3}q^{2/3}N_1^{1/3}-\delta} + \sqrt{2^{-1}3^{1/3}q^{2/3}N^{1/3}-\delta} \right)  \nonumber \\
 & \leq \frac{3^{1/6}}{2^{3/2}} q^{1/3} (N^{1/6} + N_1^{1/6}) 
 - \frac{1}{2^{3/2} \cdot 3^{1/6}} \delta q^{-1/3} (N^{-1/6} + N_1^{-1/6})\,. \label{est3}
\end{align}
The combination of Estimates \eqref{est1}, \eqref{est2} and \eqref{est3} proves the lemma.
\end{proof}

\begin{lemma} \label{lm4}
For all $N \geq \lfloor f(Z_1)\rfloor+1=N_1+1=48$, we have that $f^{-1}(N) \geq \dfrac{N}{\beta} - k\beta^{-4/3} N^{1/3}$. 
\end{lemma}

\begin{proof}
After setting $N=f(Z)$ for some $Z > Z_1$, the assertion of the lemma is equivalent to 
$$ 
Z \geq \dfrac{f(Z)}{\beta} - k\beta^{-4/3} f(Z)^{1/3}\,.
$$
But this immediately follows from the Definition \eqref{currentf} of $f$ by using the obvious estimate $f(Z)^{1/3} \geq (\beta Z)^{1/3}$. 
\end{proof}

Now, we are in the position to further estimate the right-hand side of Estimate \eqref{Ineq1}. 
Using Estimates \eqref{sigmaest} and \eqref{elementaryest} together with Lemma \ref{lm4}, we obtain that for all $N \geq N_1+1=48$,
\begin{align}
&\left(\eta N - \frac{1}{\beta} - f^{-1}(N) \right) \sqrt{\sigma_N}
 \leq \left( \left( \eta - \frac{1}{\beta} \right) N + k \beta^{-4/3} N^{1/3} - \frac{1}{\beta} \right) \sqrt{2^{-1}3^{1/3} q^{2/3} N^{1/3} - \delta(N_1,q)} \nonumber \\
& \qquad \qquad \leq \left( \left( \eta - \frac{1}{\beta} \right) N + k \beta^{-4/3} N^{1/3} - \frac{1}{\beta} \right) 2^{-1/2} \left(  3^{1/6} q^{1/3} N^{1/6} -  3^{-1/6}  q^{-1/3} \delta(N_1,q)  N^{-1/6} \right) \nonumber \\
& \qquad \qquad = 2^{-1/2} 3^{1/6} q^{1/3} \left( \eta - \frac{1}{\beta} \right)
N^{7/6} -  2^{-1/2} 3^{-1/6}  q^{-1/3}  \left( \eta - \frac{1}{\beta} \right) \delta(N_1,q)  N^{5/6} \nonumber\\
& \qquad \qquad \quad + 2^{-1/2} 3^{1/6} k \beta^{-4/3} q^{1/3} N^{1/2} -  
 2^{-1/2} \left( 3^{-1/6}  k \beta^{-4/3} q^{-1/3} \delta(N_1,q) + 3^{1/6} q^{1/3} \beta^{-1}  \right)  N^{1/6} \nonumber\\& \qquad \qquad \quad +  2^{-1/2} 3^{-1/6} q^{-1/3} \beta^{-1} \delta(N_1,q) N^{-1/6}\,. \label{est4}
\end{align}
Using this together with Lemma \ref{lm3} and Estimate \eqref{sigmaest} to bound the right-hand side of Estimate \eqref{Ineq1}, we find that
\begin{align}
L(N,\eta N) 
 &\leq \left[2^{-1/2} 3^{1/6} q^{1/3} \left( \eta - \frac{1}{\beta} \right) + \dfrac{2^{1/2}\cdot 3^{7/6}}{7\beta}  q^{1/3}  \right] N^{7/6} 
 \nonumber\\ & \quad \; - \left[ \frac{2^{1/2}\cdot 3^{5/6}}{5\beta} + \frac{1}{2^{1/2} \cdot 3^{1/6}}   \left( \eta - \frac{1}{\beta} \right) \right] q^{-1/3}  \delta(N_1,q) \, N^{5/6} \nonumber\\
& \quad \;  + 2^{-1/2} 3^{1/6} k \beta^{-4/3} q^{1/3}\, N^{1/2} 
 \nonumber\\ & \quad \; - \left[  \frac{1}{2^{1/2} \cdot 3^{1/6}}  k \beta^{-4/3} q^{-1/3} \delta(N_1,q) + \frac{3^{1/6}}{2^{3/2}} q^{1/3} \beta^{-1}  \right] \, N^{1/6}  \nonumber\\ & \quad \; + \left[  L(N_1,Z_1) - \frac{1}{\beta} \left(
  \dfrac{2^{1/2}\cdot 3^{7/6}}{7}  q^{1/3} N_1^{7/6} - \frac{2^{1/2}\cdot 3^{5/6}}{5} \delta(N_1,q) q^{-1/3} N_1^{5/6} \right.  \right. \nonumber\\ & \quad \; \left. \left. -\frac{3^{1/6}}{2^{3/2}} q^{1/3}  N_1^{1/6}  + \frac{1}{2^{3/2} \cdot 3^{1/6}} \delta(N_1,q) q^{-1/3} N_1^{-1/6} \right) \right]   \nonumber\\ & \quad \; + \frac{1}{2^{3/2} \cdot 3^{1/6}} \beta^{-1}  \delta(N_1,q) q^{-1/3} \, N^{-1/6} \qquad \qquad (N \geq N_1+1=48) \,. \label{Eest}
\end{align}
After inserting the explicit choice of the parameters from the beginning of this section, namely $q=2$, $N_1=47$, $Z_1=30.9$, $\beta=1.1185$, $k=4$, and using the bounds $L(N_1,Z_1)\leq 50.01$ and $\eta \geq 1/\beta$, this becomes
\begin{align}
L(N,\eta N)  &\leq \alpha(\eta, \beta)\eta \, N^{7/6} - 0.2405\, N^{5/6} + 3.6861 \,N^{1/2} - 1.2506\, N^{1/6} \nonumber \\&\quad\; -16.3892 + 0.1003\, N^{-1/6} \qquad \qquad (N \geq N_1+1=48) \,. \label{Eestexplicit}
\end{align}

\noindent\textit{Proof of Theorem \ref{theo2}.}
\noindent We have to show that Estimate \eqref{Elower}
holds for all $\eta \in [1/\beta,\infty)$ and all $Z \in (0,\eta N]$, where $\alpha(\eta, \beta)$ is defined by Eq. \eqref{Defalpha}. 
First, we show that it suffices to consider the endpoint case $Z=\eta N$. For this purpose, we  again consider the unitary scaling $U_Z \in \LL(\mathcal{H})$ defined by Eq. \eqref{DefU_Z}.
Then a straightforward computation, using the decomposition of $H_N(Z)$  from Eq. \eqref{H_N(Z)decomp} with $T(N)$, $V(N)$ and $I(N)$ defined by Eq. \eqref{DefTVW}, shows that $(U_Z)^{*}H_N(Z) U_Z = Z^2 \widetilde{H}_N(Z)$, where
\begin{align}
\widetilde{H}_N(Z):= T(N) - V(N) + \frac{1}{Z}\, I(N)\,. \label{DefHTilde}
\end{align}
It follows that the ground-state energy of $\widetilde{H}_N(Z)$ is given by
\begin{align*}
 \inf \sigma\left(\widetilde{H}_N(Z) \right) = \frac{E(N,Z)}{Z^2} \,. 
\end{align*}
As the interaction energy $I(N)$ is positive and has a prefactor of $1/Z$ in the Definition \eqref{DefHTilde} of $\widetilde{H}_N(Z)$, we conclude that $ \inf \sigma(\widetilde{H}_N(Z))$ is monotonically decreasing as a function of $Z>0$. This implies that $E(N,Z)/Z^2 \geq  E(N,\eta N)/(\eta N)^2$ for all $Z \in (0, \eta N]$, so it is indeed sufficient to consider the endpoint case $Z=\eta N$ in Estimate \eqref{Elower}.

To prove Estimate \eqref{Elower} for $Z=\eta N$ with arbitrary $\eta \in [1/\beta,\infty)$, we are going to apply Theorem \ref{theo1}. Note that it is then sufficient to show 
that for all $N \in \N$ and all $\eta \in [1/\beta,\infty)$,
\begin{align}
L(N, \eta N) \leq \alpha(\eta, \beta)\, \eta\, N^{7/6}\,. \label{Lboundtoshow}
\end{align}
For $N \geq 3601$, an elementary computation shows that
 \begin{align*}
&- 0.2405\, N^{5/6} + 3.6861 \,N^{1/2} - 1.2506\, N^{1/6} -16.3892 + 0.1003\, N^{-1/6} \\
&\qquad\leq N^{1/2}(- 0.2405\,N^{1/3} + 3.6861) < 0\,,
 \end{align*}
which means that Estimate \eqref{Lboundtoshow} then immediately follows from Estimate \eqref{Eestexplicit}. 
To prove that Estimate \eqref{Lboundtoshow} also holds for $N \leq 3600$, we first consider the left endpoint case $\eta=1/\beta$, where $\alpha(1/\beta, \beta)=2^{5/6}3^{7/6}/7$. 
By computing $L(N, N/\beta )$ explicitly for all $N \leq 3600$ with the help of a computer, we see that the desired bound
$L(N, N/\beta ) \leq 2^{5/6}3^{7/6} N^{7/6} /(7 \beta)$ indeed holds for any $N \leq 3600$. 
For general $\eta \in [1/\beta,\infty)$, it follows from the 
Definition \eqref{DefL} of $L(N, \eta N)$ with $\eta N \geq N/\beta > f^{-1}(N)$ that
\begin{align*}
L(N, \eta N) &= L(N, N/\beta) + \sqrt{\sigma_N} \left( \eta - \dfrac{1}{\beta} \right) N\,.
\end{align*}
Now, note that it follows from Lemma \ref{lmA2} that $\sigma_N < (3/2)^{1/3} N^{1/3}$, so we conclude that for all $N \leq 3600$,
\begin{align*}
L(N, \eta N) & \leq L(N, N/\beta) + \left( \dfrac{3}{2} \right)^{1/6} \left( \eta - \dfrac{1}{\beta} \right) N^{7/6} 
\\
& \leq \left(\dfrac{2^{5/6}\cdot 3^{7/6}}{7\eta \beta}    + \left(\dfrac{3}{2}\right)^{1/6}\left( 1 - \frac{1}{\eta \beta} \right)  \right)\, \eta\, N^{7/6}  \\
&= \alpha(\eta, \beta) \, \eta\, N^{7/6}\,,
\end{align*}
where the last line just follows from the Definition \eqref{Defalpha} of $ \alpha(\eta, \beta)$. 
This completes the proof of Estimate \eqref{Lboundtoshow}, and hence also the proof of Estimate \eqref{Elower}. Finally, the bound \eqref{IonEstimate}  in the statement of Theorem \ref{theo2} follows from Estimate \eqref{Elower} by choosing $\eta=1$ and the bound \eqref{AtomEstimate} is just \eqref{IonEstimate} in the special case $Z=N \in \N$.  \hfill$\Box$\bigskip\newline

\appendix

\section*{Appendix}

\section{Estimate of \texorpdfstring{$\sigma_N$}{sigmaN}} \label{secsigma_N}

In this section we derive an upper bound for $\sigma_N$ defined by Eqs. \eqref{DefSigmaN} and \eqref{Defh_N} in (Case FER) with arbitrary spin number $q \in \N$ (see Lemma \ref{lmA2} below). In this case it is well-known that $\sigma_N$ is the sum of the absolute values of the first $N$ eigenvalues of the one-particle operator $h_1=-\frac{1}{2}\Delta - 1/|x|$ with domain $D(h_1)=H^2(\R^3) \otimes \C^q$.
This means that
\begin{align}
\sigma_N = \sum_{k=1}^N |\lambda_k|,  \label{SigmaFer}
\end{align}
where $\lambda_k$ denotes the $k$-th eigenvalue of $h_1$ (counted with multiplicity). Since $h_1$ has precisely the eigenvalues $-1/(2n^2)\, (n \in \N)$ with corresponding multiplicity $qn^2$, it follows that $\lambda_k=-1/(2n^2)$ if and only if
$k \in (S(n-1), S(n)]$, where 
\begin{align*}
S(n):= q \sum_{k=0}^{n} k^2  \qquad \qquad (n \in \N_0). \end{align*}
In particular, it follows that for all $n_0 \in \N$,
\begin{align*}
 \sum_{k=1}^{S(n_0)} |\lambda_k| =  \sum_{n=1}^{n_0} \dfrac{1}{2n^2} \cdot qn^2 = \dfrac{q}{2}n_0\,.
\end{align*}
Now, note that $S$ can be continued to a homeomorphism $S: [0,\infty) \rightarrow [0,\infty)$ by setting
\begin{align*}
S(t):= q  \left( \frac{1}{3} t^3 + \frac{1}{2}t^2 + \frac{1}{6}t \right) \qquad \qquad (t \in [0,\infty))\,. 
\end{align*}
Then, we obtain that
\begin{align*}
 \sigma_N &= \sum_{k=1}^N |\lambda_k| \\
 &= \sum_{k=1}^{S(\lfloor S^{-1}(N) \rfloor)} |\lambda_k| + \sum_{k=S(\lfloor S^{-1}(N) \rfloor)+1}^N |\lambda_k| \\
 &= \dfrac{q}{2} \lfloor S^{-1}(N) \rfloor + \sum_{k=S(\lfloor S^{-1}(N) \rfloor)+1}^N \dfrac{1}{2(\lfloor S^{-1}(N) \rfloor +1)^2}\,,
\end{align*}
which yields the explicit formula
\begin{align}
\sigma_N = \dfrac{q}{2} \lfloor S^{-1}(N) \rfloor + \dfrac{N-S(\lfloor S^{-1}(N) \rfloor)}{2(\lfloor S^{-1}(N) \rfloor +1)^2}\,. \label{SigmaFormula}
\end{align}
Moreover, by using Cardano's formula to solve the equation $S(t)=N$, we obtain the explicit formula $S^{-1}(N)= \alpha(N) - \frac{1}{2}$ for $N \geq q/(36 \sqrt{3})$,  where 
\begin{align}
\alpha(N):= \sqrt[3]{\dfrac{3N}{2q} + \sqrt{\dfrac{9N^2}{4q^2} - \dfrac{1}{1728}} } + \sqrt[3]{\dfrac{3N}{2q} - \sqrt{\dfrac{9N^2}{4q^2} - \dfrac{1}{1728}} }\,. \label{DefS^-1}
\end{align}

\begin{lemma} \label{lmA1}
For $N, N_1 \in \N$ with $N \geq N_1 \geq  q/(36 \sqrt{3})$, we have
\begin{align*}
\alpha(N) < \left( \dfrac{3}{q} \right)^{1/3} N^{1/3} + \dfrac{1}{3^{4/3} \cdot 4} q^{1/3} \left( 1+r(N_1)^{1/3}\right) N^{-1/3},
\end{align*}
where
\begin{align}
r(N):= \dfrac{1}{2} - \dfrac{1}{2} \sqrt{1- \dfrac{q^2}{3888N^2}}\,. \label{Defr}
\end{align}
\end{lemma}
\begin{proof}
With $\xi(N):= \sqrt{\dfrac{9N^2}{4q^2}- \dfrac{1}{1728}}$, a direct computation yields
\begin{align}
\alpha(N)^3 &= \dfrac{3N}{2q} + \xi(N)
+ 3  \sqrt[3]{\left( \dfrac{3N}{2q} + \xi(N) \right) \left( \dfrac{3N}{2q} - \xi(N) \right) }\, \alpha(N) + \dfrac{3N}{2q} - \xi(N)
\nonumber \\ &= \dfrac{3N}{q} + \dfrac{1}{4} \alpha(N) \,. \label{alphapower}
\end{align}
Note that
\begin{align*}
\sqrt[3]{\dfrac{3N}{2q} +  \sqrt{\dfrac{9N^2}{4q^2}- \dfrac{1}{1728}}  } < \sqrt[3]{\dfrac{3N}{q}} \qquad \qquad (N \geq q/(36 \sqrt{3})
\end{align*}
and
\begin{align*}
\dfrac{3N}{2q} -  \sqrt{\dfrac{9N^2}{4q^2}- \dfrac{1}{1728}} 
= \dfrac{3N}{q} r(N)\,. 
\end{align*}
Hence, it follows from Eq. \eqref{DefS^-1} and the fact that $r$ is strictly monotonically decreasing in $N$ that
\begin{align*}
\alpha(N) < \sqrt[3]{\dfrac{3N}{q}} \left( 1+ r(N_1)^{1/3} \right) \qquad \qquad (N \geq N_1)\,.
\end{align*}
Using this to estimate the right-hand side of Eq. \eqref{alphapower}, we obtain for $N \geq N_1$ the upper bound
\begin{align*}
\alpha(N) &= \sqrt[3]{\dfrac{3N}{q} + \dfrac{1}{4} \alpha(N)}\nonumber \\
&< \sqrt[3]{\dfrac{3N}{q} + \dfrac{1}{4} \left( \dfrac{3}{q} \right)^{1/3}  \left( 1+ r(N_1)^{1/3} \right) N^{1/3} }  \nonumber \\
& \leq \left( \dfrac{3}{q} \right)^{1/3}  N^{1/3} + \dfrac{ \dfrac{1}{4} \left( \dfrac{3}{q} \right)^{1/3}  \left( 1+ r(N_1)^{1/3} \right) N^{1/3}}{3\left( \dfrac{3N}{q} \right)^{2/3}  } \nonumber \\
& = \left( \dfrac{3}{q} \right)^{1/3} N^{1/3} + \dfrac{1}{3^{4/3} \cdot 4} q^{1/3} \left( 1+r(N_1)^{1/3}\right) N^{-1/3}, \nonumber
\end{align*}
where the elementary estimate
\begin{align*}
(a+\delta)^{1/3} \leq  a^{1/3} + \dfrac{\delta}{3 a^{2/3}}  \qquad \qquad (a > 0, \delta \geq 0)  
\end{align*}
was used for the third line. 
\end{proof}

\begin{lemma} \label{lmA2}
For all $N, N_1 \in \N$ with $N \geq N_1 \geq  q/(36 \sqrt{3})$, we have
\begin{align*}
\sigma_N < \dfrac{3^{1/3} q^{2/3}}{2} N^{1/3} - \dfrac{q}{4} + \dfrac{1}{3^{4/3} \cdot 8} q^{4/3}  \left( 1+r(N_1)^{1/3}\right) N^{-1/3}\,,
\end{align*}
where $r$ is defined by Eq. \eqref{Defr}.
\end{lemma}
\noindent\textit{Remark.} The asymptotic behavior of $\sigma_N$ as $N \to \infty$ is well-known (see e.g. \cite{LiebStability} after Eq. (40), p. 560). However, we are not aware of a reference for the remainder estimate. 

\begin{proof}
After writing $\lfloor S^{-1}(N) \rfloor = S^{-1}(N) - \kappa(N) $ with $\kappa(N) \in [0,1)$, Eq. \eqref{SigmaFormula} becomes
\begin{align}
\sigma_N = \dfrac{q}{2} (S^{-1}(N) - \kappa(N)) + \dfrac{N-S(S^{-1}(N) - \kappa(N))}{2(S^{-1}(N) - \kappa(N) +1)^2}\,,  \label{SigmaFormula2}
\end{align}
where
\begin{align*}
N-S(S^{-1}(N) - \kappa(N)) &= S(S^{-1}(N))-S(S^{-1}(N) - \kappa(N)) \\
&= q \left( \dfrac{S^{-1}(N)^3 - (S^{-1}(N) - \kappa(N))^3}{3}  
+ \dfrac{S^{-1}(N)^2 - (S^{-1}(N) - \kappa(N))^2}{2} \right.\\
&\quad + \dfrac{S^{-1}(N) - (S^{-1}(N) - \kappa(N))}{6}\Bigg)\ \\
&= q \left( \left( S^{-1}(N) + \frac{1}{2}\right)^2 - \left( S^{-1}(N) + \frac{1}{2}\right) \kappa(N) + \frac{\kappa(N)^2}{3} - \frac{1}{12} \right) \kappa(N) \nonumber \\
&\leq q \left( \left( S^{-1}(N) + \frac{1}{2}\right)^2 - \left( S^{-1}(N) + \frac{1}{2}\right) \kappa(N) + \frac{\kappa(N)}{4}^2 \right) \kappa(N) \nonumber \\
&= q \left(S^{-1}(N) + \frac{1}{2} - \frac{\kappa(N)}{2}  \right)^2 \kappa(N)\\
& \leq q \left(S^{-1}(N) - \kappa(N) +1  \right)^2 \kappa(N)\,.
\end{align*}
Here, both inequalities used that $\kappa(N) \in [0,1)$. Using this to estimate the right-hand side of Eq. \eqref{SigmaFormula2}, we obtain that 
\begin{align*}
\sigma_N \leq \dfrac{q}{2} S^{-1}(N) \qquad \qquad (N \in \N).
\end{align*}
Now, the desired estimate follows from Lemma \ref{lmA1} since $S^{-1}(N)=\alpha(N)-  \frac{1}{2}$  for $N \geq q/(36 \sqrt{3})$.
\end{proof}

\section{Critical charge number for one bound electron}
 Although from a physical point of view it seems to be obvious  that the number of electrons that a nucleus can bind is monotonically increasing in its charge number $Z>0$, we are not aware of a rigorous proof for this. But the following lemma shows that there exists a critical charge number $Z_c \in [1/2,1)$, so that for $Z \in (0,Z_c]$ exactly one electron can be bound to  the nucleus, while for $Z>Z_c$ at least two electrons can be bound.

\begin{lemma} \label{lmZc}
Consider (Case FER) with $q=2$ and define 
\begin{align*}
Z_c:=\sup\left\{Z>0 \, : \, N_c(Z)=1 \right\} \,,
\end{align*}
where the maximum number of bound electrons $N_c(Z)$ is defined by Eq. \eqref{DefNc}. Then $Z_c \in [1/2,1)$, $N_c(Z)=1$ if  $Z \in  (0,Z_c]$ and $N_c(Z) \geq 2$ if $Z>Z_c$. In particular, $E(2,Z)=E(1,Z)=-Z^2/2$ for all $Z \in (0,Z_c]$ and $E(2,Z)<E(1,Z)$ for all $Z>Z_c$.
\end{lemma}

\noindent\textit{Remarks.} 
\begin{enumerate}[label=\arabic*.)]
\item The lemma is well-known in the literature, see for example the discussion in \cite{Hundertmark2021}.
\item Since $Z_c$ plays an important role in quantum chemistry, great effort has been made to determine the value of $Z_c$ to very high precision with the help of  numerical methods
(see e.g. \cite{Baker1990, Estienne2014, Stillinger1966}). 
The extensive numerical computations in \cite{Estienne2014} suggest that $Z_c = 0.91102822407725573(4)$.
\end{enumerate}

\begin{proof}
We consider $h(Z):=E(2,Z)/Z^2$ for $Z>0$.
Then it follows from Lemma \ref{lm01} that $h$ is continuous and it is also not difficult to show that $h$ is monotonically decreasing. Indeed, applying the unitary scaling $U_{\lambda}$ from Eq. \eqref{DefU_Z} with  $\lambda=Z$ to the decomposition of $H_2(Z)$ from Eq. \eqref{H_N(Z)decomp}, we find that
\begin{align*}
(U_{Z})^{*} H_2(Z) U_{Z} = Z^2 \left( T(2) - V(2) + \dfrac{1}{Z} I(2)\right)\,, 
\end{align*}
and hence
\begin{align*}
h(Z)= \dfrac{1}{Z^2} \inf \sigma\left(H_2(Z) \right) 
=\dfrac{1}{Z^2} \inf \sigma\left( (U_{Z})^{*} H_2(Z) U_{Z} \right) =\inf \sigma\left(T(2) - V(2) + \dfrac{1}{Z} I(2) \right)\,, \end{align*}
where the right-hand side is monotonically decreasing in $Z>0$ since $I(2) \geq 0$. 
Next, we note that Lieb's bound \eqref{LiebBound} implies that 
$N_c(Z) = 1$ for all $Z \in (0,1/2]$ because $N_c(Z) \in \N$ for all $Z>0$. In view of the Definition \eqref{DefNc} of $N_c(Z)$, this means that $h(Z)=E(2,Z)/Z^2=E(1,Z)/Z^2=-1/2$ for all $Z \in (0,1/2]$. 
Moreover, Bethe \cite{Bethe} has shown that $\textup{H}^{-}$ exists and hence $h(1)=E(2,1)<E(1,1)=-1/2$. All these properties of $h$ now motivate the definition 
\begin{align*}
Z_0:=\sup\left\{Z>0 \, : \, h(Z) = -\dfrac{1}{2} \right\}\,.  
\end{align*}
From the above, it then follows that $Z_0 \in [1/2,1)$ and we are going to show that $Z_0=Z_c$. Indeed, if $Z>Z_0$, then we have that $E(2,Z)/Z^2 = h(Z) < h(Z_0) = -1/2 = E(1,1)$ due to the definition of $Z_0$ and the properties of $h$. This is equivalent to $E(2,Z)<E(1,1)Z^2=E(1,Z)$, which means that $N_c(Z) \geq 2$ by the Definition \eqref{DefNc} of $N_c(Z)$.
If $Z \leq Z_0$, then, on the one hand, Lieb's bound and $Z_0<1$ imply that $N_c(Z) \leq 2$. On the other hand,  
it follows from the definition of $Z_0$ and the properties of $h$ that $h(Z)=-1/2$ and hence $E(2,Z)=-Z^2/2=E(1,Z)$.
Altogether, we conclude that $N_c(Z)=1$ if and only if $Z \leq Z_0$ and hence
\begin{align*}
Z_0=\sup\left\{Z>0 \, : \, N_c(Z)=1 \right\}=Z_c \,,
\end{align*}
which completes the proof.
\end{proof}

\section{Monotonicity of the Ionization Energies} \label{secionization}
Experiments and numerical calculations suggest that the energy that is required to remove an electron from an atom or ion increases as more and more electrons have already been removed. More precisely, this means that the ionization energies and electron affinities are monotonic in the sense that for all $n \leq N-1$,
\begin{align}
E(N-n,N) - E(N-n+1,N) \leq E(N-n-1,N) - E(N-n,N)\,, 
\label{IonizationIneqClaim}
\end{align}
where equality holds if and only if the right-hand side is zero. Here, for $n \in \N$, the difference $E(N-n,N) - E(N-n+1,N)$  corresponds to the $n$-th ionization energy of the atom with charge number $N$, and for  $n \leq 0$ it corresponds to the $(|n|+1)$-th electron affinity.
Note that Estimate \eqref{IonizationIneqClaim} can be interpreted as a convexity condition for $E(\,\cdot\,,N)$ for any fixed $N \in \N$. Physicists justify Estimate  \eqref{IonizationIneqClaim} by arguing that the screening of the nuclear charge decreases as more and more electrons are removed \cite{Briet}, but a rigorous proof  of Estimate \eqref{IonizationIneqClaim} within the Schrödinger theory of quantum mechanics still remains an open problem \cite{SimonProblems}. However, for the lower bound $E_L(\,\cdot\,,Z)$ from Theorem \ref{theo1} even a slightly stronger version of Estimate \eqref{IonizationIneqClaim} can be proven:
\begin{theorem}\label{theo4}
Let the hypotheses of Theorem \ref{theo1} be satisfied in either (Case FER) or (Case FS). Then, for all $N \geq N_0+2$ and all $Z \geq Z_0$,
\begin{align}
E_L(N-1,Z) - E_L(N,Z)  \leq E_L(N-2,Z) - E_L(N-1,Z)\,.\label{IonizationIneq}
\end{align}
\end{theorem}

\begin{proof}
If $Z \in [Z_0, f^{-1}(N)]$, then it follows from the recursive Definition \eqref{DefL} of $L(\,\cdot \,,Z)$ that $L(N,Z)=L(N-1,Z)$ and hence $E_L(N-1,Z)-E_L(N,Z)=0$. Moreover, since $N \mapsto E_L(N,Z)$ is monotonically decreasing and hence $E_L(N-2,Z) - E_L(N-1,Z) \geq 0$, Estimate \eqref{IonizationIneq} is satisfied in this case. 

Let us now consider the more difficult case that $Z>f^{-1}(N)$ for some $N \geq N_0+2$. Due to the strict monotonicity of $f$, we then have that 
$Z>f^{-1}(N) > f^{-1}(N-1)>f^{-1}(N-2)$. 
Therefore, the definition of $L(\,\cdot \,,Z)$ in Theorem \ref{theo1} implies that
\begin{align}
L(N-2,Z)&=L_0 + a_{N-1}\sqrt{\sigma_{N-2}}\,,  \label{EqN-2}\\
L(N-1,Z)&=L_0 + a_{N-1}\sqrt{\sigma_{N-1}}\,, \label{EqN-1} \\
L(N,Z)&=L(N-1,f^{-1}(N))+ a_N\sqrt{\sigma_{N}} =L_0 + b_N\sqrt{\sigma_{N-1}} + a_N\sqrt{\sigma_{N}} \,,  \label{EqN}
\end{align}
where $L_0:= L(N-2,f^{-1}(N-1))$, $a_{N-1}:=Z-f^{-1}(N-1)$, $a_{N}:=Z-f^{-1}(N)$ and $b_{N}:= f^{-1}(N)-f^{-1}(N-1)$ for short. Eqs. \eqref{EqN-2}-\eqref{EqN} imply that  
\begin{align}
&E_L(N,Z)+E_L(N-2,Z)-2E_L(N-1,Z) \nonumber \\
&\qquad=2L(N-1,Z)^2 - L(N,Z)^2 - L(N-2,Z)^2 \nonumber \\
&\qquad= 2L_0 \left( 2 a_{N-1} \sqrt{\sigma_{N-1}} - a_{N-1} \sqrt{\sigma_{N-2}} - a_N \sqrt{\sigma_{N}} - b_N \sqrt{\sigma_{N-1}} \right) \nonumber \\
&\qquad \quad + 2 (a_{N-1})^2 \sigma_{N-1} -(a_{N-1})^2 \sigma_{N-2}  - (a_N)^2 \sigma_N  - (b_N)^2 \sigma_{N-1} - 2 a_N b_N \sqrt{\sigma_{N}} \sqrt{\sigma_{N-1}}\,.  \label{ELEq}
\end{align}
To prove Estimate \eqref{IonizationIneq}, we have to show that the right-hand side of Eq. \eqref{ELEq} is non-negative. For this purpose, we first recall that $\sigma_N$ is defined by Eq. \eqref{DefSigmaN} with $h_N$ given by Eq. \eqref{Defh_N}. As already mentioned after Eq. \eqref{Defh_N}, we have the explicit formula $\sigma_N=N/2$ in (Case FS), whereas in (Case FER) $\sigma_N$ is given by Eq. \eqref{SigmaFer}. In any case, it is straightforward to verify that for all $N \geq 3$,
\begin{align}
2 \sigma_{N-1} - \sigma_{N-2} -  \sigma_N \geq 0\,.  \label{SigmaConcave}
\end{align}
Using this together with $\sigma_{N-1} \leq \sigma_N$ and the identity $a_N+b_N=a_{N-1}$, we obtain the estimate
\begin{align*}
&2 (a_{N-1})^2 \sigma_{N-1} -(a_{N-1})^2 \sigma_{N-2}  - (a_N)^2 \sigma_N  - (b_N)^2 \sigma_{N-1} - 2 a_N b_N \sqrt{\sigma_{N}} \sqrt{\sigma_{N-1}} \\
& \qquad \geq 2 (a_{N-1})^2 \sigma_{N-1} -(a_{N-1})^2 \sigma_{N-2}  - (a_N+b_N)^2 \sigma_N \\
& \qquad=  (a_{N-1})^2 \left(2 \sigma_{N-1} - \sigma_{N-2} -  \sigma_N \right) \\ & \qquad  \geq 0\,.
\end{align*}
This shows that the second term on the right-hand side of Eq. \eqref{ELEq} is non-negative. For the first term we use the following estimate:
\begin{align*}
&2 a_{N-1} \sqrt{\sigma_{N-1}} - a_{N-1} \sqrt{\sigma_{N-2}} - a_N \sqrt{\sigma_{N}} - b_N \sqrt{\sigma_{N-1}} \\
& \qquad \geq  2 a_{N-1} \sqrt{\sigma_{N-1}} - a_{N-1} \sqrt{\sigma_{N-2}} - (a_N+b_N) \sqrt{\sigma_{N}} \\
& \qquad = a_{N-1} \left(2 \sqrt{\sigma_{N-1}} -  \sqrt{\sigma_{N-2}} - \sqrt{\sigma_{N}} \right)\,. 
\end{align*}
Therefore, to show that the right-hand side of Eq. \eqref{ELEq} is non-negative, it only remains to show that for all $N \geq 3$, 
\begin{align*}
 \sqrt{\sigma_{N-1}}  \geq \dfrac{\sqrt{\sigma_{N-2}} + \sqrt{\sigma_N}}{2}\,.
\end{align*}
However, this easily follows from Estimate \eqref{SigmaConcave} and the concavity of the square root function. 
\end{proof}

\noindent\large\textbf{Acknowledgement}\newline
\normalsize The author thanks Ioannis Anapolitanos and Dirk Hundertmark for many helpful discussions and for proofreading parts of this work. The research for this article was funded by the Deutsche Forschungsgemeinschaft (DFG, German Research Foundation) – Project-ID 258734477 – SFB 1173. 

\bibliographystyle{plain}
\bibliography{LowerBoundsBib}

\end{document}